\documentclass{article}
\usepackage[utf8]{inputenc}
\usepackage{amsmath,amsfonts,amsthm, amssymb, xcolor}
\usepackage{mathtools}
\usepackage{graphicx}
\usepackage{tikz}
\usetikzlibrary{arrows}
\usepackage{float}
\usepackage{caption}
\usepackage{subcaption}

\usepackage{algorithm}
\usepackage[noend]{algpseudocode}
    \input xy

    \xyoption{all}

\newcommand{\bb}[1]{\mathbb{#1}}

\newcommand{\vect}[1]{\mathbf{#1}}
\newcommand{\dope}[1]{\mathrm{dope}({#1})}
\newcommand{\cdope}[1]{\mathrm{cdope}({#1})}

\usepackage{etoolbox}

\newtheorem{lemma}{Lemma}[section]

\newtheorem{proposition}{Proposition}[section]
\theoremstyle{definition}
\newtheorem{definition}{Definition}[section]

\newcommand{\Mod}[1]{\ (\mathrm{mod}\ #1)}

\title{The DOPE Distance is SIC: A Stable, Informative, and Computable Metric on Time Series And Ordered Merge Trees}

\author{Christopher J. Tralie, Zachary Schlamowitz, Jose Arbelo, \\Antonio I. Delgado, Charley Kirk, Nicholas A. Scoville}

\begin{document}
\maketitle
\begin{abstract}

Metrics for merge trees that are simultaneously stable, informative, and efficiently computable have so far eluded researchers.  We show in this work that it is possible to devise such a metric when restricting merge trees to ordered domains such as the interval and the circle.  We present the ``dynamic ordered persistence editing'' (DOPE) distance, which we prove is stable and informative while satisfying metric properties.  We then devise a simple $O(N^2)$ dynamic programming algorithm to compute it on the interval and an $O(N^3)$ algorithm to compute it on the circle.  Surprisingly, we accomplish this by ignoring all of the hierarchical information of the merge tree and simply focusing on a sequence of ordered critical points, which can be interpreted as a time series.  Thus our algorithm is more similar to string edit distance and dynamic time warping than it is to more conventional merge tree comparison algorithms.  In the context of time series with the interval as a domain, we show empirically on the UCR time series classification dataset that DOPE performs better than bottleneck/Wasserstein distances between persistence diagrams.

\end{abstract}

\section{Introduction}
\label{sec:intro}

\begin{figure}[h]
    \centering
    \includegraphics[width=\textwidth]{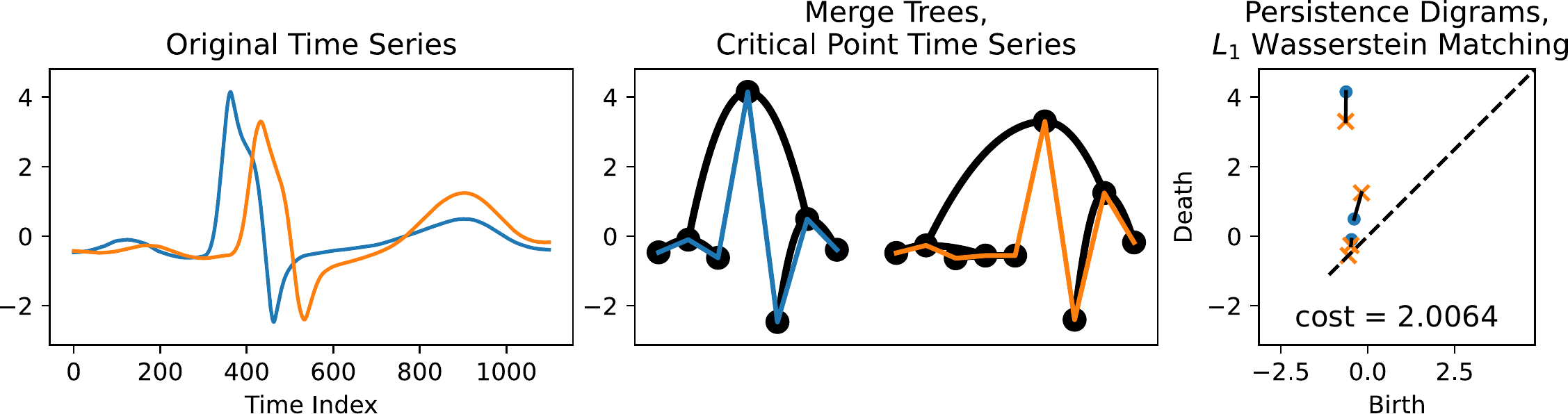}
    \caption{An example of two smoothed, z-normalized ECG signals from the UCR dataset \cite{dau2019ucr}.  These are considered to be part of the same class, but no uniform rescaling can align them perfectly. On the other hand, topological signatures, such as merge trees and persistence diagrams that summarize them, are blind to parameterization, though merge trees can be difficult to compare.}
    \label{fig:ECGConcept}
\end{figure}

\begin{figure}
    \centering
    \includegraphics[width=\textwidth]{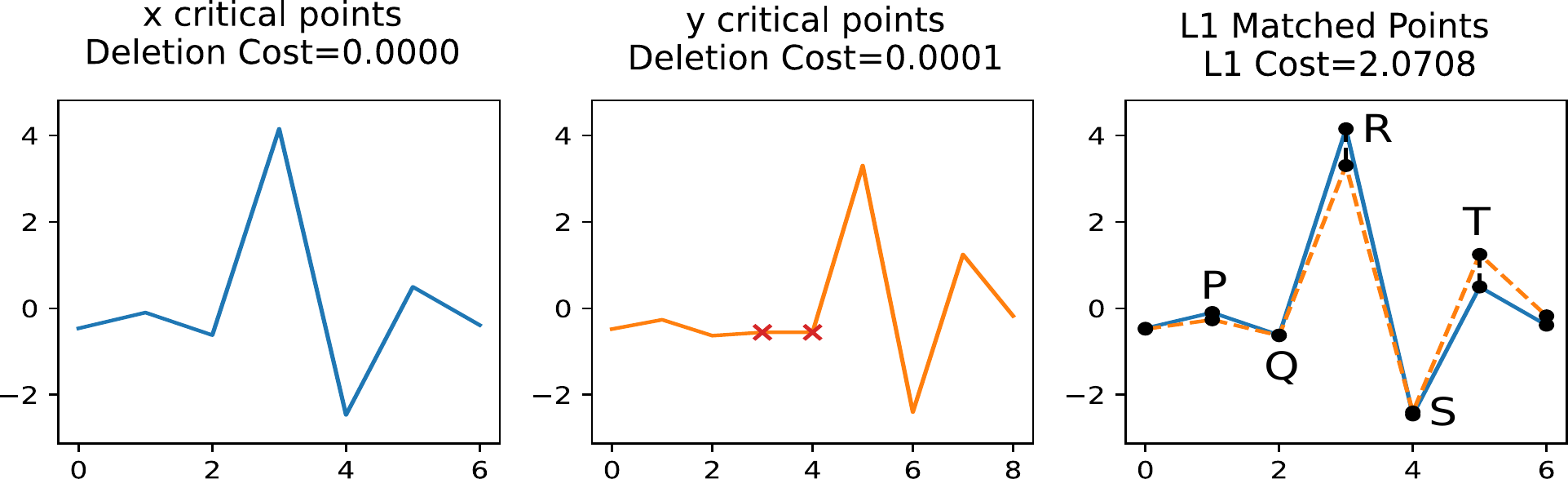}
    \caption{By contrast to Figure~\ref{fig:ECGConcept}, DOPE matching throws away the tree hierarchy and instead compares the ``critical point time series'' directly.  In this example, the critical points that DOPE matches pick up on the characteristic PQRST ECG pattern of a heartbeat.  Note also that the Wasserstein matching cost (Figure~\ref{fig:ECGConcept}) does not exceed the DOPE cost; this is a general ``informativity'' property that we prove in Section~\ref{sec:informativity}.}
    \label{fig:ECGConceptDOPE}
\end{figure}

One dimensional time series comparison is a ubiquitous problem, with application domains in such disparate fields as comparing eye track traces \cite{fang2018electrooculography}, ECG heartbeat analysis \cite{olszewski2001generalized}, earthquake analysis \cite{arul2021applications}, lightning strike categorization \cite{eads2002genetic}, phoneme recognition \cite{hamooni2014dual}, activity recognition via motion capture, and many more (see, for instance, the UCR time series classification database \cite{dau2019ucr}, which has 128 time series datasets).  One challenge that occurs across many domains is that the time series in the same class may be ``warped,'' or subject to re-parameterizations more complicated than uniform shifts or scales, hampering retrieval using sample to sample Euclidean distance comparisons.  Figure~\ref{fig:ECGConcept} shows an example of this in ECG readings of a heartbeat.  There are myriad techniques in the literature which are blind to or insensitive to such warps (Section~\ref{sec:background}), but none of them simultaneously satisfy a desirable set of theoretical properties that we outline in Section~\ref{sec:desiderata}.

In this work, we define a parameterization blind dissimilarity measure between 1D time series that we dub ``Dynamic Ordered Persistence Editing (DOPE)'' that simultaneously satisfies polynomial time computability, metric properties, stability, and informativity, all of which we define precisely in Section~\ref{sec:desiderata}.  We arrive at DOPE by focusing on techniques that build a topological structure known as a merge tree (otherwise known as a ``join tree'' \cite{carr2003computing} or a ``barrier tree'' \cite{flamm2002barrier}) on time series.  These trees organize a hierarchy of critical points, which is preserved under re-parameterization.  DOPE instead ignores the hierarchy completely and does an edit distance on the critical points as a {\em sequence} implied by traversing this hierarchy in order.  Since it is an edit distance, we easily infer a mapping from critical points of one time series to another.  Moreover, an advantage of this over more traditional ordered time series comparison tools, such as dynamic time warping \cite{sakoe1970similarity,sakoe1978dynamic}, is that it gives a parsimonious and intuitive representation of the matching, since it's restricted to critical points.

Beyond time series, DOPE is also a metric for merge trees on the interval and on circles; in fact, comparing (circular) time series and comparing ordered merge trees are equivalent with DOPE.  By contrast, most techniques for comparing merge trees have very general domains (Section~\ref{sec:bgmergetree}) which make it difficult to satisfy our four desired properties simultaneously.

\subsection{Time Series And Merge Trees}

We now define time series, critical points, and merge trees in a manner that suits our purposes, and we elucidate a connection between time series and merge trees.  

\begin{definition}
\label{def:timeseries}
A \textbf{1D time series} is a time-ordered sequence of $N$ numbers $\mathbf{x} = [\vect{x}_0, \vect{x}_1, \ldots , \vect{x}_{N-1}]$.  A \textbf{circular 1D time series} is an equivalence class of sequences with an equivalence relation given by a circular shift; that is, given a representative time series $\vect{x}$, it is the equivalence class $\{  \overset{\rightarrow k}{\vect{x}} \}_{k=0}^{N-1}$, where $\overset{\rightarrow k}{\vect{x}} = [\vect{x}_{k \mod N}, \vect{x}_{1+k \mod N}, \ldots , \vect{x}_{N-1+k \mod N}]$.

\end{definition}

1D time series can be thought of as an ordered collection of samples from a continuous function on a topological interval domain, while 1D circular time series can be thought of as samples from a continuous function on a topological circular domain.

\begin{definition}
    \label{def:critpoint}
    A \textit{critical point} is a local extremum of a time series.  More specifically, for interior points ($1 < i < N$) of a 1D time series and all circular points, $x_i$ is a {\em local min} if $x_i < x_{i-1 \mod N}, x_i < x_{i+1 \mod N}$ and $x_i$ is a {\em local max} if $x_i > x_{i-1 \mod N}, x_i > x_{i+1 \mod N}$.  As a special case, the endpoints of a time series are only considered critical points if they are less than their adjacent neighbor, in which case they are local mins.  
\end{definition}

The special case is to help with a correspondence to merge trees, as we explain below.  Now that we have critical points, we can form a time series out of them alone.

\begin{definition}
We associate to any time series $\mathbf{x}$ its \textbf{critical times series} $\mathbf{x}^c$ given by removing all non-critical (regular) points of $\mathbf{x}$ and reindexing according to the order of the remaining points of $\mathbf{x}^c$.  We further define $\vect{x}^{cm}_i$ to be an \textit{indicator function} that determines if $\vect{x}^c_i$ is a min or max, i.e., $\vect{x}^{cm}_i = -1$ if $\vect{x}^c_i$ is a min and $\vect{x}^{cm}_i = 1$ if $\vect{x}^c_i$ is a max.   
\end{definition}

The critical point time series is a sparse set of important information about a 1D function, and it includes everything needed to specify the function up to a parameterization.  Due to the Euler characteristic $\chi$, every time series has an odd number of critical points (starting and ending with a min), and every circular time time series has an even number of them.

Critical points of a scalar function on some domain can be organized into a hierarchical structure known as a \textbf{merge tree}.  The general definition is as follows:

\begin{definition}
\label{def:mergetree}
Given a scalar function $f: X \to \mathbb{R}$, define a \textbf{sublevelset} of this function as $S_f(t) := \{ (x,t) \in X \times \mathbb{R}: f(x) \leq t \}$.  Consider the following equivalence relation: $(x, y) \in S_f(\infty) \sim (x', y) \in S_f(\infty)$ if and only if there there is a path $\gamma: [0, 1] \to S_f(y)$ so that $\gamma(0)=(x, y)$ and $\gamma(1)=(x', y)$.  Then the \textbf{merge tree} associated to $f$ is the quotient space of $S_f(\infty)$ under this equivalence relation.
\end{definition}

Intuitively, a sublevelset is the pools of water that form with the graph of a function as a basin, and a merge tree tracks the different pools of water at different heights $y$.  When $y$ passes through a local min, a new pool is ``born,'' and a leaf node exists in a tree.  When $y$ passes through a (generic) saddle point, two pools merge and one of them ``dies'' and merges to the other, and there is an internal node.  In the case where the domain $X$ is $\mathbb{R}$, this is quite simple; there is a leaf node for each local min and an internal node for each local max.  Thus, we henceforth refer to min-max pairs rather than min-saddle pairs.

We can associate to each merge tree a summary known as a \textbf{persistence diagram} $\text{DGM} = \{ (b_1, d_1), (b_2, d_2), \hdots, (b_n, d_n) \}$, which consists of pairs of mins and maxes that correspond to the creation, or ``birth,'' of a component in $S_f(t)$ at height $b$ and the merging of that component to another, or a ``death,'' at height $d$, respectively.  The so-called ``elder rule'' pairs to each max the min associated to the most recently born connected component out of the two.  The other min becomes the representative for the merged component.

\begin{figure}
    \centering
    \includegraphics[width=\textwidth]{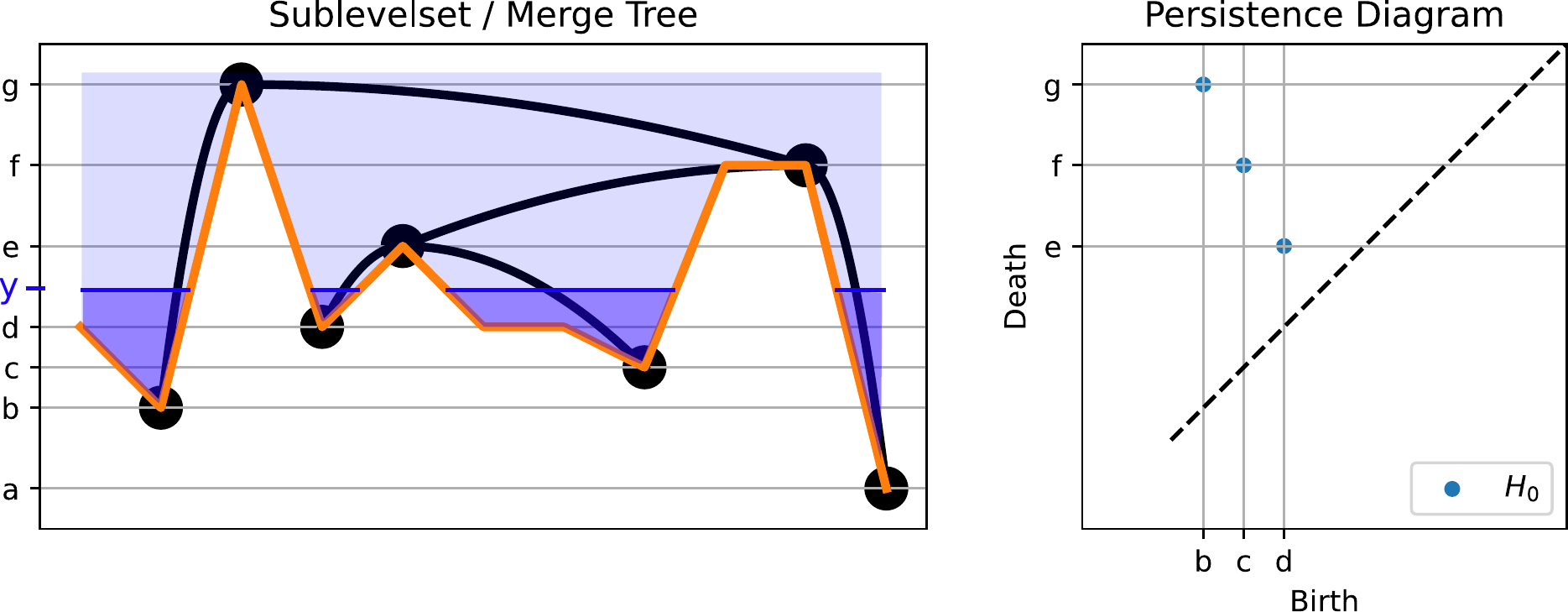}
    \caption{An example of a merge tree (drawn in black) on a piecewise linear extension of a time series (draw in orange).  The sublevelset $S_f(\infty)$ is drawn in light blue, while the sublevelset $S_f(y)$ is drawn in darker blue.  There are 4 equivalence classes in the merge tree at height $y$.  Note also that only the min on the right endpoint of the time series shows up in the merge tree; the max on the left endpoint does not lead to any births or deaths of connected components in the sublevelset.}
    \label{fig:MergeTreeBasics}
\end{figure}

To bridge the gap between the continuous formulation in Definition~\ref{def:mergetree} and our discrete notion of time series, we define a time series merge tree as that obtained on a piecewise linear extension of the time series.  Figure~\ref{fig:MergeTreeBasics} shows an example.  We also note the following:

\begin{lemma}
The critical point time series corresponds to an inorder traversal of the merge tree of a piecewise linear extension of that time series.
\end{lemma}

Therefore, {\em any matching between the critical point time series corresponds to a matching between ordered merge trees}.  Furthermore, Figure~\ref{fig:MergeTreeBasics} demonstrates the interval endpoint cases in Definition~\ref{def:critpoint}. Only endpoint mins appear in the critical point time series; no components of the sublevelset are born or die at maxes, so there are no nodes in the merge tree for them.

Finally, we introduce our first comparison between merge trees by comparing their persistence diagrams via the Bottleneck and $p$-Wasserstein distances \cite{cohen2010lipschitz}:

\begin{definition}
\label{def:wass}
\begin{equation}
d_W^p(DGM_1, DGM_2) = \inf_{\gamma \in \Gamma} \left( \sum_{i=1}^{|\gamma|} || DGM_1(i), \gamma(DGM_1(i)) ||^p \right)^{\frac{1}{p}}
\end{equation}
where $\Gamma$ is the set of all perfect bipartite matchings (1-1 correspondences between points in $DGM_1$ to those in $DGM_2$), where diagonal points $(b, b)$ are included in each with infinite multiplicity.  If $p = \infty$, this is known as the \textbf{bottleneck distance}.
\end{definition}

\section{Desired Properties}
\label{sec:desiderata}

Below are the four properties that we want to satisfy for some dissimilarity measure $\mathrm{d}(\vect{x}, \vect{y})$ between two time series $\vect{x}$ and $\vect{y}$, respectively, and perhaps also in relation to a third time series $\vect{z}$.  We draw inspiration for these properties from Morozov et. al. \cite{morozov2013interleaving}.  While they define them for merge trees on general domains, we specialize them to time series.

\begin{definition}

A \textbf{pseudometric} $d$ between time series $\mathbf{x}$ and $\mathbf{y}$ is a real-valued function satisfying:   
\begin{itemize}
    \item $\mathrm{d}(\vect{x}, \vect{x}) = 0$; that is, the distance between a time series and itself should be 0, but we also allow distance between two different time series to be 0, which is common for other dissimilarity measures in the literature (e.g. \cite{morozov2013interleaving})
    \item $\mathrm{d}(\vect{x}, \vect{y}) \geq 0$ Non-negativity
    \item $\mathrm{d}(\vect{x}, \vect{y}) = \mathrm{d}(\vect{y}, \vect{x})$ Symmetry
    \item $\mathrm{d}(\vect{x}, \vect{y}) \leq \mathrm{d}(\vect{x}, \vect{z}) + \mathrm{d}(\vect{z}, \vect{y})$ for all time series $\mathbf{z}$; Triangle inequality
\end{itemize}
\end{definition}

The bottleneck and Wasserstein metrics between persistence diagrams of sublevelset filtrations of these functions are a pseudometric, as is the interleaving distance between merge trees \cite{morozov2013interleaving}.  Satisfying this property can help accelerate search in large databases \cite{ciaccia1997m}.

Next, given critical time series $\vect{x^c}$ of length $M_c$ and $\vect{y^c}$ of length $N_c$ associated to $\vect{x}$ and $\vect{y}$, respectively, we define \textbf{discrete $p$-stability} between $\vect{x}$ and $\vect{y}$ as

\begin{definition}
    \label{def:stable}
    A metric $\mathrm{d}(\vect{x}, \vect{y})$ is \textbf{discrete $p$-stable} if 
    
    \begin{equation}
        \mathrm{d}(\vect{x}, \vect{y}) \leq  \left( \sum_{i = 1}^{\max{M_c, N_c}} |\vect{x}^c_i - \vect{y}^c_i| ^p \right) ^{1/p}, p \in \mathbb{Z}^+
    \end{equation}
    
\end{definition}

where we zeropad so that $\vect{x}_i^c = 0, i > M_c, \vect{y}_j^c = 0, j > N_c$.

%Define notion of stability for time series and apply it to Wasserstein distance between persistence diagrams and merge tree edit distance

Beyond metric properties and stability, we want our metrics to be strong enough so that they are at least as strong as the $p$-Wasserstein distance between persistence diagrams of sublevelset filtrations of $\vect{x}$ and $\vect{y}$.  This is inherently at odds with stability.  More formally,

\begin{definition}
    A metric $\mathrm{d}(\vect{x}, \vect{y})$ is \textbf{$p$-informative} if 
    \begin{equation}
        \mathrm{d}^p_{\mathcal{W}}(\mathrm{DGM}(\vect{x}), \mathrm{DGM}(\vect{y})) \leq \mathrm{d}(\vect{x}, \vect{y})
    \end{equation}
\end{definition}
where $\mathrm{DGM}(\vect{v})$ refers to the persistence diagram of the sublevelset filtrations on a time series $\vect{v}$.  For instance, if $p = \infty$, then we require our metric to be at least as strong as the bottleneck distance between diagrams.  Note also that if $p' < p$, then $p'$-informativity is a stronger condition than $p$-informativity for persistence diagrams of functions on cell complexes, of which (circular) time series are a special case \cite{skraba2020wasserstein}.

Finally, we want there to exist an algorithm to \textbf{efficiently compute} the metric $\mathrm{d}(\vect{x}, \vect{y})$ in asymptotic polynomial time as a function of the lengths of $\vect{x}$ and $\vect{y}$.  This property has been surprisingly elusive for merge tree metrics that satisfy the first three properties.  For instance, the authors of \cite{agarwal_computing_2018} show that the interleaving distance, which satisfies the first three properties for $p=\infty$, is NP-hard to approximate within a constant factor less than $3$.  Our goal is to exploit the additional ordered structure to restrict possibilities.

\section{Related Work}
\label{sec:background}

We now review some related work for comparing time series and merge trees.  We focus primarily on parameter free, unsupervised techniques, many of which work beyond ordered domains.  Interestingly, though, each technique fails at least one of our desired properties.

\subsection{Merge Tree Comparison Techniques}
\label{sec:bgmergetree}

\begin{figure}
    \centering
    \includegraphics[width=\textwidth]{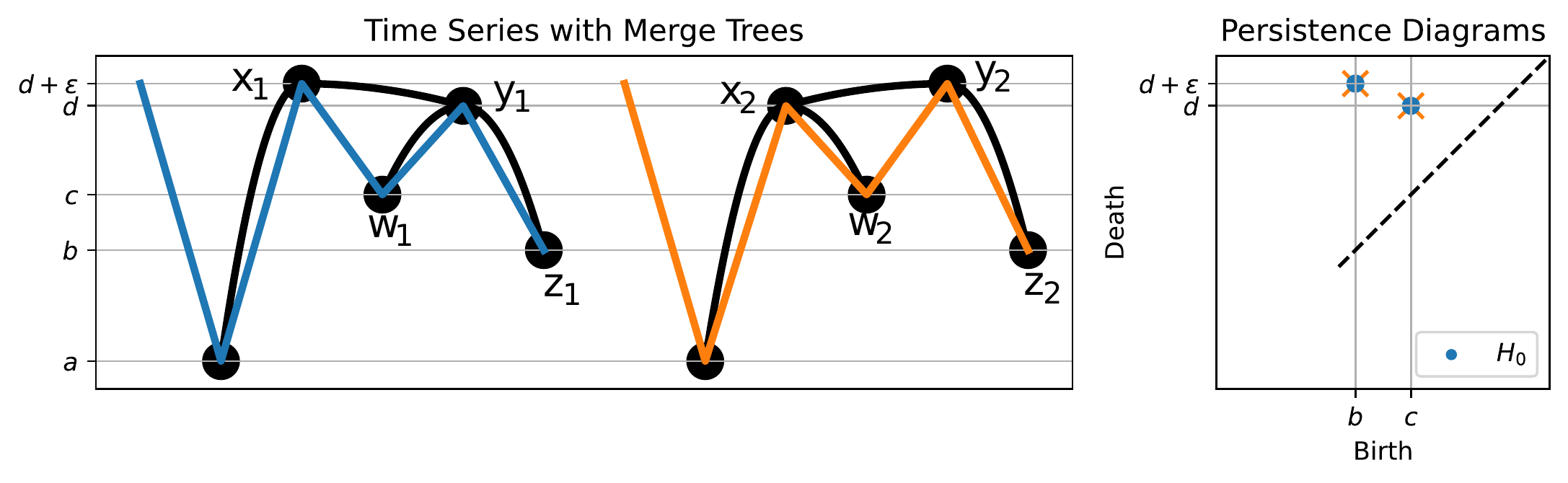}
    \caption{These two time series only differ at the maxes $x$ and $y$ by $\epsilon$, and their DOPE distance is appropriately $2\epsilon$, but their merge trees do not preserve ancestral relationships, so the best map that merge tree edit distance \cite{sridharamurthy_edit_2020} has to delete the mins $w_1$ and $w_2$ before matching the min $z_1$ to $z_2$, violating stability.  By contrast, their persistence diagrams are identical, so any metrics based on persistence diagrams fail to see the difference between them for any $0 \leq \epsilon \leq d-c$.}
    \label{fig:BlindRotation}
\end{figure}

\begin{figure}
\centering
\includegraphics[width=\textwidth]{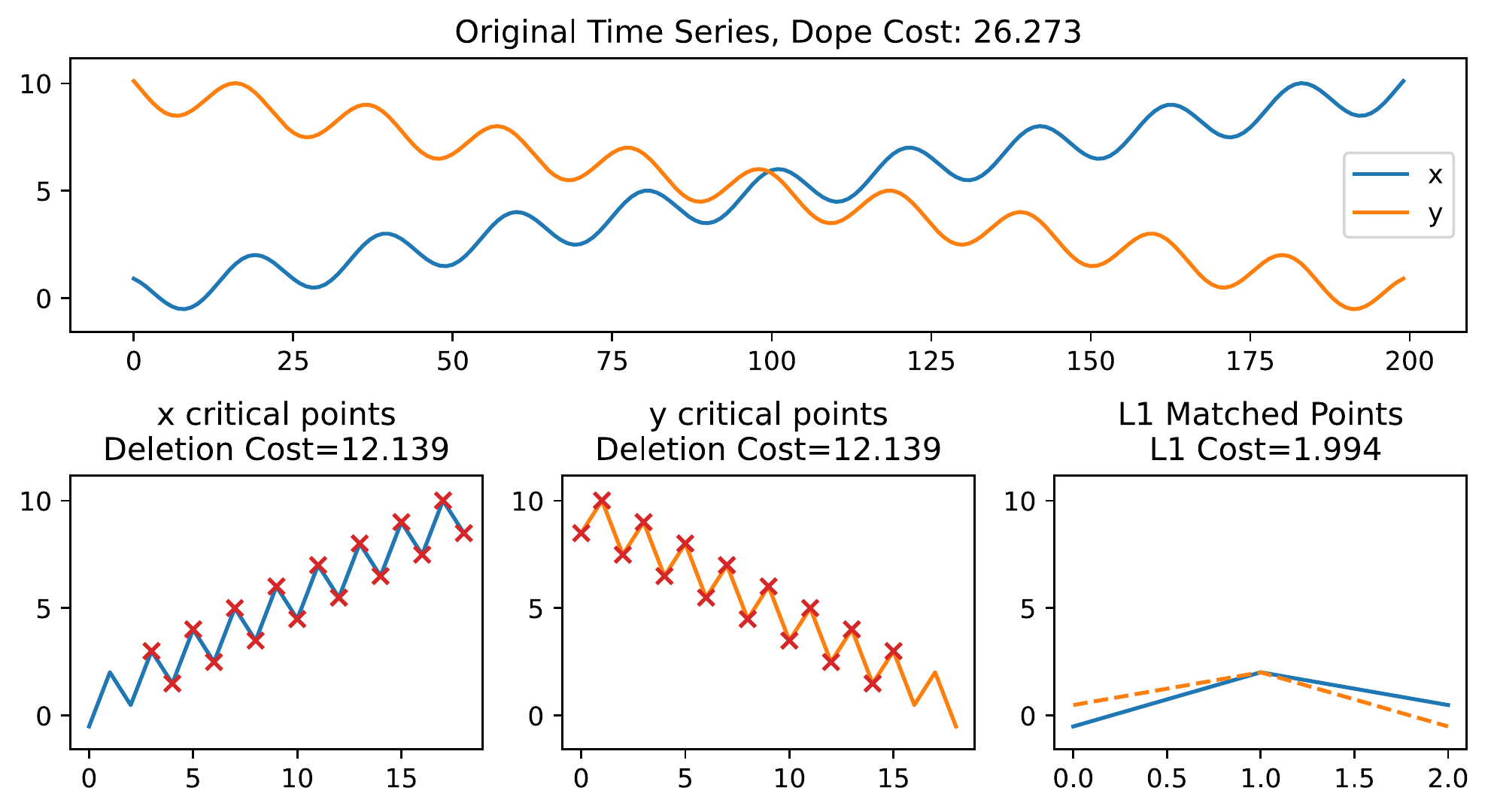}
\caption{These two time series differ by a reflection, and their persistence diagrams are identical.  On the other hand, DOPE must match critical points in {\em time-ordered sequence}, so it is forced to delete many pairs before matching, and it can hence tell the difference between these two time series.}
\label{fig:reflect}
\end{figure}

Recent work defining a universal metric (i.e. maximally informative while remaining stable) for merge tree comparison \cite{cardona_universal_2022} seems to imply that it will be difficult to find an algorithm to efficiently compute such a metric.  Nevertheless, many works have found compromises.  Weaker distances use statistics on top of persistence diagrams (e.g. \cite{rouse2015feature,khasawneh2018topological,dindin2020topological,myers2020automatic}), which can confuse very different time series (as seen in Figure~\ref{fig:BlindRotation} and Figure~\ref{fig:reflect}), but which are stable and polynomial time computable.

% Node-based edit distance algorithms for trees include operations of deleting, inserting, and relabeling tree nodes, with relabeling costs left open to the application.  The edit distance is then defined to be the minimum sum of the aforementioned costs over all sequences of editing operations that transform one tree into another.

% Because the ancestral ordering has changed, a natural map which maps $w_1$ to $w_2$ and $z_1$ to $z_2$ is not allowed.  Instead, the merge tree edit distance is forced to delete the branches $(w_1, y_1)$ and $(w_2, y_2)$, which violates stability.

One stronger family of techniques to compare merge trees is based on general node-based edit distance algorithms for trees.  If the costs for add/delete/match satisfy the triangle inequality, then such edit distances will be metrics \cite{zhang_simple_1989}.  However, some desirable edit distances are hard to compute; for instance, Zhang shows \cite{zhang1992editing} that ancestor-preserving maps for unordered trees are NP complete even for an alphabet of size 2.  In spite of this, researchers have shown some empirical success with edit-based operations.  Sridharamurthy et. al. \cite{sridharamurthy_edit_2020} plug in matching costs that are based on metrics between birth/death pairs, such as the $L_{\infty}$ distance.  They then use a suboptimal but polynomial time computable constrained version of the ancestral mapping edit distance \cite{zhang1996constrained} which maps disjoint subtrees to disjoint subtrees.  However, even with a map that preserves ancestral orders, stability is lost with tree rotations as maxes move past each other, as shown in Figure~\ref{fig:BlindRotation}.  The authors deal with this instability in practice by collapsing saddle points below some threshold $\epsilon$ to a degenerate saddle; in this example, this amounts to deleting nodes $y_1$ and $y_2$ before matching.  Still, the authors show good performance on 2D domains, and it is possible to improve performance by extending the costs to incorporate branch decompositions \cite{bremer_measuring_2014} of merge trees \cite{pont_wasserstein_2021, WetzelsLG22}.

Beyond edit distance, another class of approaches to building metrics between merge trees are “functional approaches.” These define a distance in terms of functions which act on the space of trees. One prominent such distance is the interleaving distance, which roughly finds the smallest vertical shift necessary to map all points on each tree upwards into the other tree. Originally defined in \cite{morozov2013interleaving} by  Morozov et al. in terms of two maps (one in each direction between two trees), the interleaving distance is generally hard to compute \cite{agarwal_computing_2018}. However, while Morozov et al.'s original presentation mentions only an exponential-time computational algorithm, later work lessens the computational burden to an $O(n^2\log(n))$ fixed-parameter tractable problem. This is achieved using a dynamic programming algorithm based on the fact that the interleaving distance can be defined equivalently by mapping only in one direction instead of finding two maps \cite{farahbakhsh2019fpt}.

Finally, there is a technique that uses integer linear programming to devise a merge tree metric \cite{pegoraro_metric_2022}.  While integer linear programming is NP hard, this is one of the few papers that implements working code to compare merge trees.  It is also worth noting that if we constrain ourselves to phylogenetic trees in the leaves are labeled, it is possible to devise stable, informative, and efficiently computable metrics \cite{gasparovic38intrinsic}.  Unfortunately, in time series, we do not have specific labels on the mins, so these techniques do not directly apply.

\subsection{Time Warping Time Series Comparison Techniques}

Dynamic Time Warping (DTW) is a classical algorithm to compare time series that are samples of re-parameterized version of the same function by finding a time-ordered correspondence known as a {\em warping path} (Definition~\ref{def:warppath}) which minimizes the sum of dissimilarities between corresponding points.  It was first developed in the context of audio recognition and alignment \cite{sakoe1970similarity, sakoe1978dynamic}, but it has seen wide applications in many time series tasks since.  DTW is efficiently computable for time series in general spaces in $O(N^2)$ time, and it is computable $O(N^2 \log \log \log(N) / \log \log(N)$ for $N$ as $O(M)$ for 1D time series \cite{gold2018dynamic}.  The authors of \cite{chung2020persistent} also claim that it is $L_{\infty}$ informative, but it is neither stable nor a metric (See Appendix~\ref{sec:appendixdtwpropeties}).  One way to turn DTW into a metric, while making it less informative, is to merely minimize the maximum dissimilarity over all possible warping paths, rather than the sum of all dissimilarities.  This is known as the Discrete Fréchet Distance \cite{eiter1994computing}, which can be computed in subquadratic $O(mn \log \log n / \log n)$ time \cite{agarwal2014computing}.

\section{Dynamic Ordered Persistence Editing (DOPE)}

\begin{figure}
\centering
\includegraphics[width=\textwidth]{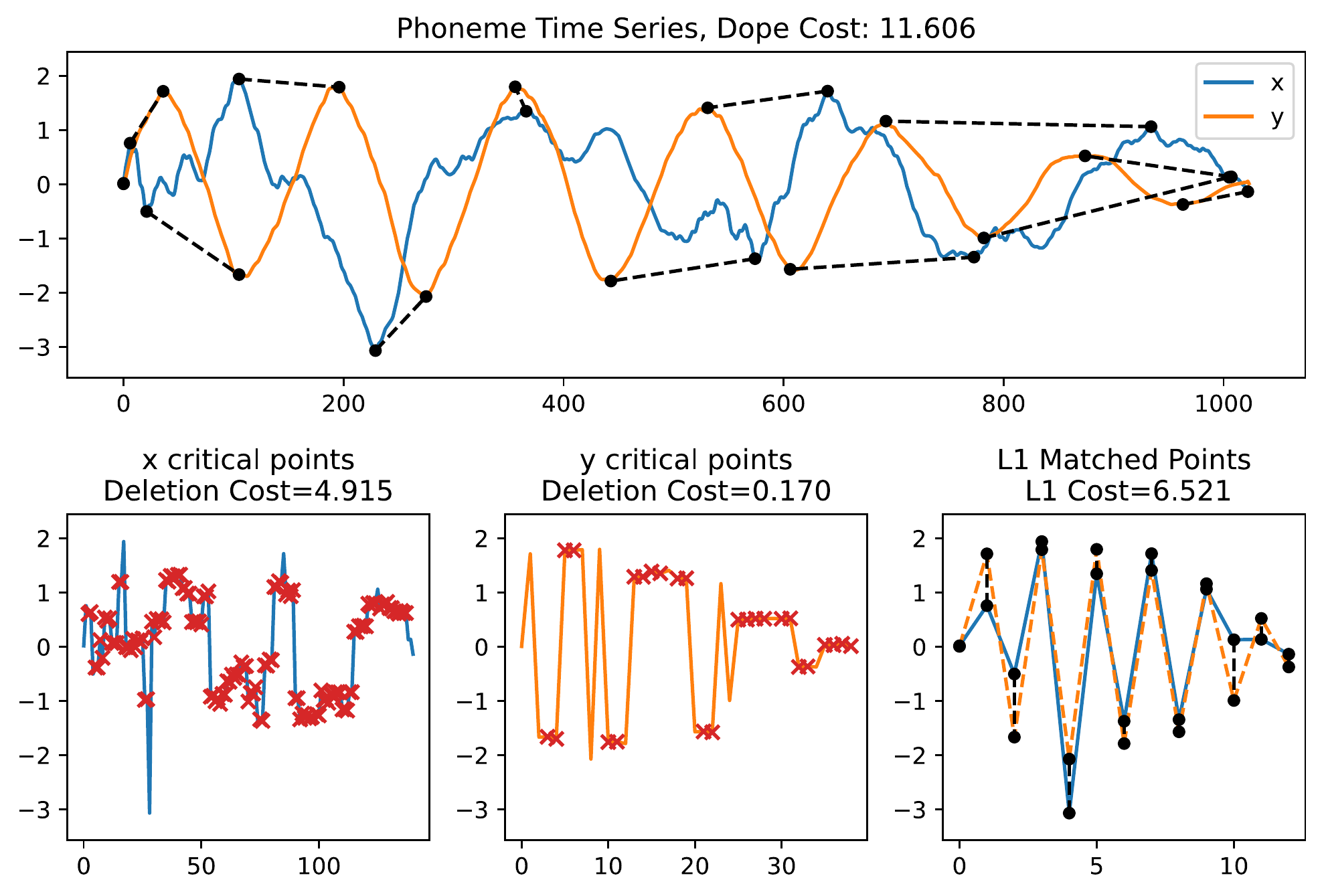}
\caption{An example of an optimal DOPE matching between the audio samples of two different people saying the phoneme "EH", as obtained from the  UCR time series database \cite{dau2019ucr}.}
\label{fig:DOPExample}
\end{figure}

We now propose our new dissimilaity measure: Dynamic Ordered Persistence Editing (DOPE). Before defining DOPE, we provide a few preliminary definitions. In what follows, all time series are assumed to be defined on the same domain.

\begin{definition}
    Given a time series $\vect{x}$, a \textbf{\textit{min-max pair}} is an ordered pair $(\vect{x}^c_k, \vect{x}^c_{k+1})$ of critical points in $\vect{x}^c$.
\end{definition}

\begin{definition} Let $\vect{x}, \vect{y}$ be time series.  An \textbf{\textit{alignment}} $A=(M, R_{\vect{x}^c}, R_{\vect{y}^c})$ of $\vect{x}$ and $\vect{y}$ is a triple of sets $M\subseteq \vect{x}^c\times \vect{y}^c, R_{\vect{x}^c}\subseteq \vect{x}^c\times \vect{x}^c, R_{\vect{y}^c}\subseteq \vect{y}^c\times \vect{y}^c$ satisfying the following properties:
\begin{itemize}
    \item $M$ contains matched pairs from $\vect{x}^c\times \vect{y}^c$ with the property that any two matchings $(\vect{x}^c_{i_1},\vect{y}^c_{j_1}), (\vect{x}^c_{i_2}, \vect{y}^c_{j_2}) \in M$ satisfying $i_1 < i_2$ also satisfy $j_1 < j_2$. 
    \item Ordered pairs in $R_{\vect{x}^c}$ and $R_{\vect{y}^c}$ are min-max pairs.
    \item Each element of $\vect{x}^c$ (respectively, $\vect{y}^c$) appears in exactly one ordered pair in $M$ or $R_{\vect{x}^c}\ (R_{\vect{y}^c})$, but not both.
    \item $\left|M\right| = \left|\vect{x}^c\setminus \{\vect{x}^c_k,\vect{x}^c_{k+1}: (\vect{x}^c_k,\vect{x}^c_{k+1})\in R_{\vect{x}^c}\}\right|=\left|\vect{y}^c\setminus \{\vect{y}^c_j,\vect{y}^c_{j+1}: (\vect{y}^c_j,\vect{y}^c_{j+1})\in R_{\vect{y}^c}\}\right|.$
\end{itemize}
\end{definition}

Ordered pairs in $M$ are called \textit{matched pairs} or \textit{matchings}. We call any ordered pair in $R_{\vect{x}^c}$ or $ R_{\vect{y}^c}$ a \textit{removed} or \textit{deleted pair}, and to include a min-max pair in $R_{\vect{x}^c}$ or $R_{\vect{y}^c}$ is to \textit{perform a removal} or \textit{deletion.} We now associate a cost with an alignment:
\begin{definition}
The \textbf{\textit{alignment cost}} of alignment $A=(M,R_{\vect{x}^c}, R_{\vect{y}^c})$ between time series $\vect{x}$ and $\vect{y}$ is:\\
\[C(A) := \displaystyle\sum_{(\vect{x}^c_i, \vect{y}^c_j) \in M} |\vect{x}^c_i - \vect{y}^c_j| + \displaystyle\sum_{(\vect{x}^c_i, \vect{x}^c_{i+1}) \in R_{\vect{x}^c}} |\vect{x}^c_i - \vect{x}^c_{i+1}| + \displaystyle\sum_{(\vect{y}^c_j, \vect{y}^c_{j+1}) \in R_{\vect{y}^c}} |\vect{y}^c_j - \vect{y}^c_{j+1}|\]
\end{definition}
We are now ready to define our main object of interest.

 %(was "Dynamic Persistence-Based Subsequence Alignment"):
\begin{definition}
\label{def:dope}
We define the \textbf{\textit{DOPE (Dynamic Ordered Persistence Edit) distance}} between two time series $\vect{x}$ and $\vect{y}$ to be: \[\dope{\vect{x}, \vect{y}} := \min \{C(A) \mid A \text{ is an alignment of } \vect{x} \text{ and } \vect{y}\}\]
\end{definition}

%IDEA: To show that one of our matchings lower bounds an arbitrary sequence of edits... doesn't that just follow from defining our distance as the infimum over all possible matchings and deriving a matching from the sequence of edits? i.e.:
\subsection{Metric Properties and Correspondence to an Edit Distance}
We now recast an alignment $A=(M,R_{\vect{x}^c},R_{\vect{y}^c})$ as a sequence of edit operations on the time series $\vect{x}^c$. We take inspiration from the edit operations of Zhang and Sasha \cite{zhang_simple_1989}. An \textit{edit operation from source time series $\vect{x}^c$ to target time series $\vect{y}^c$} is any one of the following: (1)  an element of $R_{\vect{y}^c}$ is an \textit{insertion}, (2) an element of $R_{\vect{x}^c}$ is a \textit{deletion}, (3) an element of $M$ is a \textit{matching}. We think of an insertion as taking a min-max pair from $\vect{y}^c$ and inserting it into $\vect{x}^c$ and of a deletion as removing a min-max pair from $\vect{x}^c$. The set of all admissible edit operations from source time series $\vect{x}^c$ to target time series $\vect{y}^c$ over all posible alignments $A$ is denoted $E_{\vect{x}^c\vect{y}^c}$.  We now define a cost function to keep track of the cost of such edits. 
\begin{definition}
    Let $\gamma \colon E_{\vect{x}^c\vect{y}^c} \to \mathbb{R}$ be the \textbf{cost function} associated with the set all of possible edit operations on $\vect{x}^c$ to $\vect{y}^c$, defined as follows:
    \begin{align*}
        &\text{One-Point Matching cost: } & \gamma( \vect{x}^c_i \rightarrow \vect{y}^c_j ) &= |\vect{x}^c_i-\vect{y}^c_j|  \\\\
        &\text{Pairwise Matching cost: } & \gamma( (\vect{x}^c_i,\vect{x}^c_{i+1}) \rightarrow (\vect{y}^c_j, \vect{y}^c_{j+1})) &= \min{\begin{cases}
        |\vect{x}^c_i-\vect{y}^c_j| + |\vect{x}^c_{i+1}-\vect{y}^c_{j+1}|, \\
        |\vect{x}^c_i-\vect{x}^c_{i+1}| + |\vect{y}^c_j-\vect{y}^c_{j+1}| \\
        \end{cases}}  \\\\
        &\text{Pairwise Deletion cost: } & \gamma( (\vect{x}^c_i,\vect{x}^c_{i+1}) \rightarrow \Lambda) &= |\vect{x}^c_i-\vect{x}^c_{i+1}|\\\\
        &\text{Pairwise Insertion cost: } & \gamma( \Lambda \rightarrow (\vect{x}^c_i,\vect{x}^c_{i+1})) &= |\vect{x}^c_i-\vect{x}^c_{i+1}| \\
    \end{align*}
    where $(\vect{x}^c_i,\vect{x}^c_{i+1}) \rightarrow \Lambda$ denotes deleting the pair $(\vect{x}^c_i,\vect{x}^c_{i+1})$ from time series $\vect{x}^c$ and $\Lambda \rightarrow (\vect{x}^c_i,\vect{x}^c_{i+1})$ denotes inserting the pair $(\vect{x}^c_i,\vect{x}^c_{i+1})$ into time series $\vect{x}^c$. When inserting the pair $(\vect{x}^c_i,\vect{x}^c_{i+1})$, the index $i$ denotes the insertion index, and all indices to the right are re-ordered. 
\end{definition}
In the case of adjacent matchings, $m_1 = (\vect{x}^c_{i_1},\vect{y}^c_{j_1}), m_2 =(\vect{x}^c_{i_2},\vect{y}^c_{j_2})$ such that $i_2 = i_1 + 1, j_2 = j_1 + 1$, we \textbf{require} the use of the pairwise matching cost.
\begin{lemma}\label{lem: triangle edit operations} \textbf{(Triangle Inequality on the Edit Operations)} Let $p \rightarrow q$, $q \rightarrow r$, and $p \rightarrow r$ denote any edits between $\vect{x}^c$ and $\vect{y}^c$, $\vect{y}^c$ and $\vect{z}^c$, and $\vect{x}^c$ and $\vect{z}^c$, respectively. Then $$\mathrm{\gamma}(p \rightarrow r) \leq \gamma (p \rightarrow q) +\gamma (q\rightarrow r).$$
\end{lemma}

\begin{proof}
    The triangle inequality on individual edits is proved by considering all possible combinations of edit operations. The following diagram illustrates eight possibilities (excluding the sequence of only deletions) by considering all paths of arrows from $\vect{x}^c$ into $\vect{z}^c$.
    $$
    \xymatrix{
    \vect{x}_i^c \ar[r] & \vect{y}_j^c  \ar[r] & \vect{z}_k^c\\
    (\vect{x}_i^c, \vect{x}_{i+1}^c)  \ar[r] \ar[rd] & (\vect{y}_j^c, \vect{y}_{j+1}^c) \ar[r] \ar[rd] & (\vect{z}_k^c, \vect{z}_{k+1}^c)\\
    \Lambda \ar[ru] \ar[r] & \Lambda \ar[ru] \ar[r] & \Lambda
    }
    $$
    Proofs for each path follow from the triangle inequality on $\bb{R}$ and are given in Appendix~\ref{sec:dopetriangleappendix}.\end{proof}

We now use Lemma \ref{lem: triangle edit operations} to prove a triangle inequality for a general sequence of edits, which will be used to prove metric properties for \textit{DOPE}. Define the cost of a sequence of edit operations as the sum of costs of the individual edits. Care must be taken as the composition of valid \textit{DOPE} edit operations need not result in a valid sequence of \textit{DOPE} edits. 

\begin{lemma}\label{compositionlemma}
    Let $\vect{x}, \vect{y}, \vect{z}$ be three time series. Given sequences of edits $S_{\vect{x^c}\vect{y}^c}$ and $S_{\vect{y}^c\vect{z}^c}$ converting $\vect{x}^c$ into $\vect{y}^c$ and $\vect{y}^c$ into $\vect{z}^c$, respectively, there exists a sequence of edits $S_{\vect{x}^c\vect{z}^c}$ converting $\vect{x}^c$ into $\vect{z}^c$ such that $$\gamma(S_{\vect{x}^c\vect{z}^c}) \leq \gamma(S_{\vect{x}^c\vect{y}^c}) + \gamma(S_{\vect{y}^c\vect{z}^c}).$$
\end{lemma}
%Notice that if edits $s_i$ and $s_j$ are disjoint (that is, the elements they modify are separate), the order in which they are performed 
\begin{proof}
    Suppose  $S_{\vect{x}^c\vect{y}^c} = s_1, \dots, s_n$ and $S_{\vect{y}^c\vect{z}^c} = s_{n+1},  \dots, s_m$ are sequences of edits converting time series $\vect{x}^c$ into $\vect{y}^c$ and $\vect{y}^c$ into $\vect{z}^c$. We compose these sequences of edits by performing the edits in succession: $S_{\vect{y}^c\vect{z}^c} \circ S_{\vect{x}^c\vect{y}^c} := s_1, \dots, s_n, s_{n+1},  \dots, s_m.$ One can track the fate of any given point in $\vect{x}^c$ through this sequence by composing those edits that involve the point. There are ten possible compositions to consider, summarized in Figure 6, falling into two categories: overlapping and non-overlapping. The eight non-overlapping edits are precisely those paths considered in the preceding lemma. Two representative overlapping edits are shown; note that interchanging $\vect{x}^c$ and $\vect{z}^c$ and reversing the directions of arrows in the schematic gives different but equivalent representatives.\\
    \begin{figure}
    \centering
    \includegraphics[width=\textwidth]{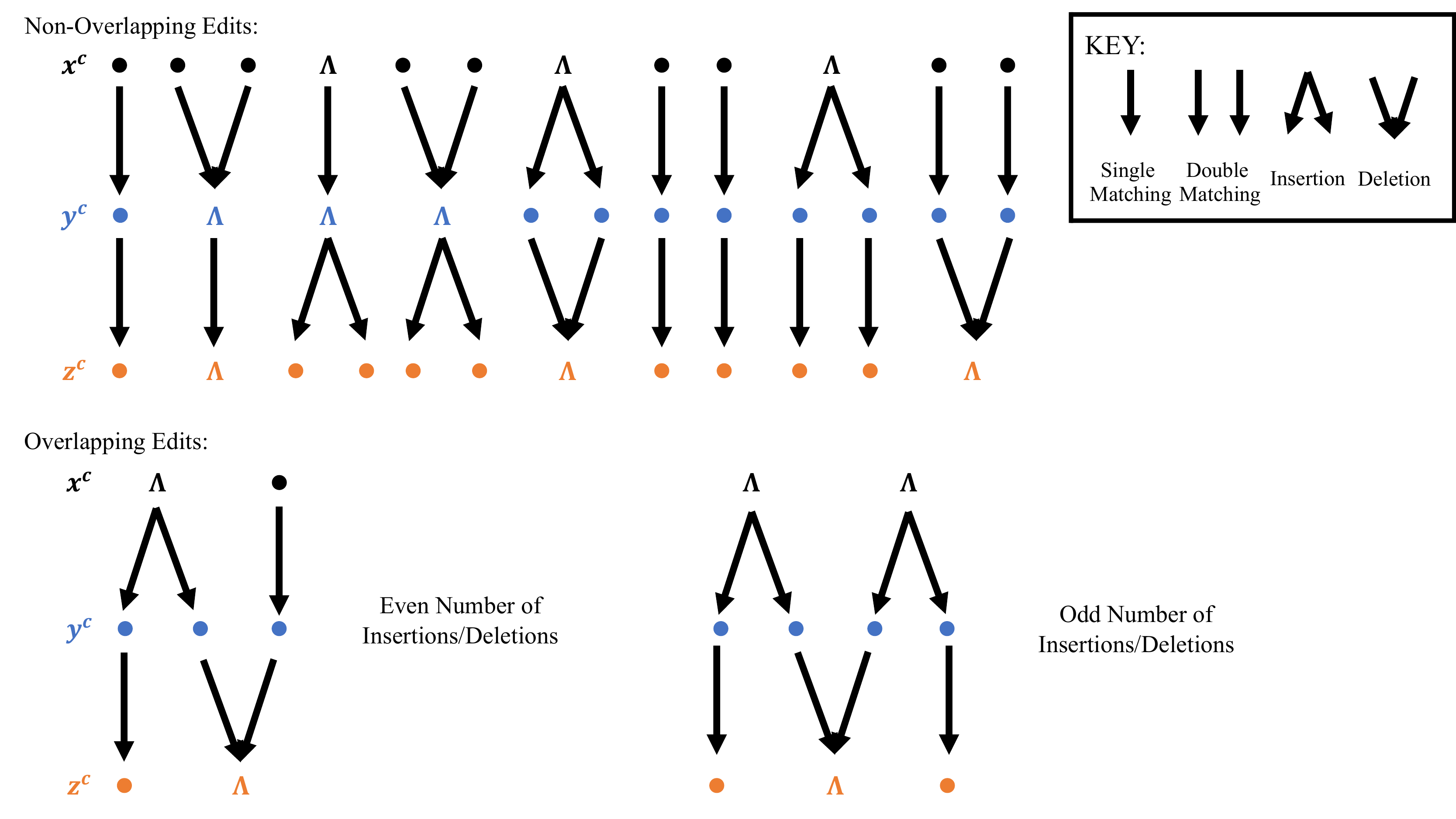}
    \caption{The possible compositions of edits to a point or min-max pair in source time series $\vect{x}^c$.}
    \label{fig:TrInCases}
    \end{figure}
    \textbf{Case 1:} All edits in $S_{\vect{y}^c\vect{z}^c} \circ S_{\vect{x}^c\vect{y}^c}$ are non-overlapping. Absent overlap, edits are disjoint, i.e., different min-max pairs are edited independently. This gives rise to a natural composition of sequential edits. For instance, two successive matchings (of either type) compose to a single matching (of that type), a pairwise matching followed by a deletion composes to a single deletion, etc. Thus $S_{\vect{x}^c\vect{z}^c} = S_{\vect{y}^c\vect{z}^c} \circ S_{\vect{x}^c\vect{y}^c}$ immediately forms a sequence of DOPE edits converting $\vect{x}^c$ into $\vect{z}^c$: $S_{\vect{x}^c\vect{z}^c} = S_{\vect{y}^c\vect{z}^c} \circ S_{\vect{x}^c\vect{y}^c} = \{s_1, \dots, s_n, s_{n+1},  \dots, s_{n+n}\} = \{ \Tilde{s_1}, \dots, \Tilde{s_n} \}$, where $\Tilde{s_i}$ is formed from the composition of exactly two edits $s_j, s_k$ with $1 \leq j,k \leq m$.\\
    \textbf{Case 2:} $S_{\vect{y}^c\vect{z}^c} \circ S_{\vect{x}^c\vect{y}^c}$ contains overlapping edits. Since edits are no longer necessarily disjoint, it is not possible to form a sequence of DOPE edits converting $\vect{x}^c$ into $\vect{z}^c$ simply by composing them $\{s_1, \dots, s_n, s_{n+1},  \dots, s_m\}$ individually. However, as summarized in Figure 6, all possible instances of overlapping edits contain either an even or odd number of insertions into and deletions from $\vect{y}^c$. In both cases we can produce a single edit from $\vect{x}^c$ to $\vect{z}^c$ that resolves the overlapping region as follows. Notice that in both cases there exist exactly two one-point matchings. In the even case, one of these matchings takes point $x$ in $\vect{x}^c$ to point $y_j$ in $\vect{y}^c$ while the other takes a different point $y_k$ in $\vect{y}^c$ to point $z$ in $\vect{z}^c$. Resolve this overlap by matching $x$ with $z$. In the odd case, the one-point matchings are either both from $\vect{x}^c$ to $\vect{y}^c$, i.e., $x_i \rightarrow y_j$, $x_{i+1} \rightarrow y_k$, $k>j$, in which case we form a deletion of min-max pair $(x_i,x_{i+1})$ from $\vect{x}^c$, or they are both from $\vect{y}^c$ to $\vect{z}^c$, i.e., $y_j \rightarrow z_l$, $y_k \rightarrow z_{l+1}$, $k>j$, in which case we form an insertion of min-max pair $(z_l,z_{l+1})$ into $\vect{z}^c$. The cost of these edits then lower bounds that of the overlapping edits; this follows from the triangle inequality on $\bb{R}$ and is shown explicitly in Appendix~\ref{sec:dopetriangleappendix}. 
    If it happens that two regions of even overlap are adjacent in the given sequences of edits, such that these resolutions produce two adjacent one-point matchings in the edits from $\vect{x}^c$ to $\vect{z}^c$, merge them into a pairwise matching.\\
    In either case, we can now form new sequence of valid DOPE edits $S_{\vect{x}^c\vect{z}^c}$ converting $\vect{x}^c$ into $\vect{z}^c$. By the triangle inequality of the preceding lemma on individual compositions of edits  and the cases above, it then follows that: $\gamma(S_{\vect{x}^c\vect{z}^c}) \leq \gamma(S_{\vect{y}^c\vect{z}^c} \circ S_{\vect{x}^c\vect{y}^c}) = \gamma(S_{\vect{y}^c\vect{z}^c}) + \gamma(S_{\vect{x}^c\vect{y}^c})$
  %  Starting at any point in $\vect{x}^c$, it is possible to track the evolution of that point through the sequence of edits; each point is ultimately matched or deleted. Often, the fate of one point (in the case of matching) or one min-max pair (in the cases of pairwise matching or deletion) is independent of that of its neighbors. For instance, consider an example [FIGURE] where a point-matching occurs next to a deletion; to understand the fate of the pair, it is unnecessary to know that the neighboring point was matched. However, in some cases, this independence is violated by means of overlap between pointwise and pairwise operations [FIGURE]. For instance, ... . This overlapping behavior ultimately takes on two forms: the number of overlapping pairwise operations is either even or odd. We consider each case individually and show that one can find a valid (non-overlapping) operation acting on the involved points of $\vect{x}^c$ (by means of matching or deletion):
%     \begin{enumerate}
%         \item Even number of pairwise operations: Here, match the existing involved point of $\vect{x}^c$ with the existing involved point of $\vect{z}$. [TODO: show algebra for triangle inequality]. \\
%         \item Odd number of pariwise operations: Here, delete the existing involved pair of points of $\vect{x}$ [TODO: show algebra for triangle inequality].\\
%     \end{enumerate}
\end{proof}

% Now, say you give me a sequence of edits $S=s_1, s_2, \dots , s_n$ that take me from time series $\vect{x}$ to time series $\vect{y}$. I claim the following:%...

%\subsection{Unused Proof Material *ignore this*}
%\text{preserving the order of the edits of the same type (I'll explain what I mean by that momentarily), I should get to y. By "preserve order of the edits of the same type" I mean the following: let's say s_i, s_j, and s_k are edits in the given sequence S, where i comes before j comes before k. Let's say i and k are deletions, but j is an insertion.  I'm going to "sort" all the edits s_i in S so that I have all the insertions first, then the "matchings", and then the deletions. However, I do this only by re-ordering edits of *different* types, so that when I'm done, I'll have s_j with the other insertions, but s_j still precedes s_k in the subsequence of deletions. Now, I have a permuted version of my original sequence S that now gives insertions first, then "matchings"/perturbations, then deletions, and this sequence of edits corresponds to paying the persistence}

%%%%%%%%%%%%%%%%%%%%%%%%%%%%%%%%%%%%%%%

% Before we prove triangle inequality for our metric, we still need to build up SOME structure -- thought previously it was the lemma about cost(composed alignment) <= cost(alignment1) + cost(alignment2), but that approach is currently under investigation for fraud

%%%%%%%%%%%%%%%%%%%%%%%%%%%%%%%%%%%%%%%

\begin{lemma}[The DOPE distance is a (pseudo)-metric]
    Let $T$ be the collection of all one-dimensional time series. Then $\mathrm{dope} \colon T \times T \rightarrow \mathbb{R}$ satisfies psuedo-metric properties. 
\end{lemma}

%That is, for any time series $\vect{x}$, $\vect{y}$ and $\vect{z}$:
%    \begin{enumerate}
%        \item $\dope{\vect{x},\vect{y}} \geq 0$ 
%        \item If $\vect{x} = \vect{y}$, then $\dope{\vect{x},\vect{y}} = 0$ 
%        \item $\dope{\vect{x},\vect{y}} = \dope{\vect{y},\vect{x}}$
%        \item $\dope{\vect{x},\vect{z}} \leq \dope{\vect{x},\vect{y}} + \dope{\vect{y},\vect{z}}$
%    \end{enumerate}

\begin{proof}
1) Non-negativity follows trivially from the definition of dope.\\
2) Suppose $\vect{x}=\vect{y}$. Then $\vect{x}^c = \vect{y}^c$, so there exists an alignment $A^*$ between $\vect{x}$ and $\vect{y}$ consisting of just matchings ($R_\vect{x} =R_{\vect{y}^c} = \emptyset$). The cost of this alignment is $C(A^*) = 0$, and non-negativity gives that $0\leq \dope{\vect{x},\vect{y}} \leq C(A^*) =0$, so $\dope{\vect{x}, \vect{y}} = 0$.  \\
3) Symmetry also follows trivially from the definition of the DOPE distance, since if $A$ is an alignment of $\vect{x}$ and $\vect{y}$ then equivalently $A$ is alignment of $\vect{y}$ and $\vect{x}$.\\
4) We now show that the triangle inequality for DOPE follows from the triangle inequality on the individual edit operations.  Let $\vect{x}, \vect{y}, \vect{z}$ be time series, and consider $\dope{\vect{x},\vect{y}}$ and $\dope{\vect{y},\vect{z}}$. By definition, there exist alignments $A_{\vect{x}\vect{y}}$ and $A_{\vect{y}\vect{z}}$ such that $C(A_{\vect{x}\vect{y}}) = \dope{\vect{x},\vect{y}}$ and $C(A_{\vect{y}\vect{z}}) = \dope{\vect{y},\vect{z}}$. From these alignments, we can produce sequences of edits $S_{\vect{x}^c\vect{y}^c}$ and $S_{\vect{y}^c\vect{z}^c}$ converting $\vect{x}^c$ into $\vect{y}^c$ and $\vect{y}^c$ into $\vect{z}^c$, respectively. Note that the order of performing these edits does not affect the total cost of the sequence, so $\gamma(S_{\vect{x}^c\vect{y}^c}) = C(A_{\vect{x}\vect{y}})$ and $\gamma(S_{\vect{y}^c\vect{z}^c}) = C(A_{\vect{y}\vect{z}})$. By Lemma \ref{compositionlemma}, there exists a sequence of edits $S_{\vect{x}^c\vect{z}^c}$ converting $\vect{x}^c$ into $\vect{z}^c$ such that $\gamma(S_{\vect{x}^c\vect{z}^c}) \leq \gamma(S_{\vect{x}^c\vect{y}^c}) + \gamma(S_{\vect{y}^c\vect{z}^c})$. From sequence $S_{\vect{x}^c\vect{z}^c}$, we form an alignment $A_{\vect{x}\vect{z}}$ of time series $\vect{x}$ with time series $\vect{z}$ by converting those operations which insert into $\vect{x}^c$ to deletions from $\vect{z}^c$. But this means there exists alignment $A_{\vect{x}\vect{z}}$ such that $C(A_{\vect{x}\vect{z}}) \leq \gamma(S_{\vect{x}\vect{y}}) + \gamma(S_{\vect{y}\vect{z}}) = \dope{\vect{x},\vect{y}} +  \dope{\vect{y},\vect{z}}$. Since $\dope{\vect{x},\vect{y}}$ is the minimum alignment cost over all alignments between $\vect{x}$ and $\vect{z}$, it follows that
$
    \dope{\vect{x},\vect{z}} \leq C(A_{\vect{x}\vect{z}}) \leq \gamma(S_{\vect{x}\vect{y}}) + \gamma(S_{\vect{y}\vect{z}}) = \dope{\vect{x},\vect{y}} +  \dope{\vect{y},\vect{z}}$\end{proof}

Note that if we restrict to only critical time series $\vect{x}^c$, then $\dope{-,-}$ becomes a bona fide metric, i.e., $\dope{\vect{x}^c,\vect{y}^c}=0$ if and only if $\vect{x}^c=\vect{y}^c$. 

%%%%%%%%%%%%%%%%%%%%%%%%%%%%%%%%%%%%%%%
\subsection{Informativity}
\label{sec:informativity}

%Recall that we may associate associate a persistence diagram to a time series by recording birth-death pairs coming from the sublevel set filtration. The Wasserstein distance between two time series is then the Wasserstein distance between their associated persistence diagrams.  See ???? for the details. We now show that our metric is more discriminating than the Wasserstein metric. 

\begin{figure}[h]
\centering
\includegraphics[width=\textwidth]{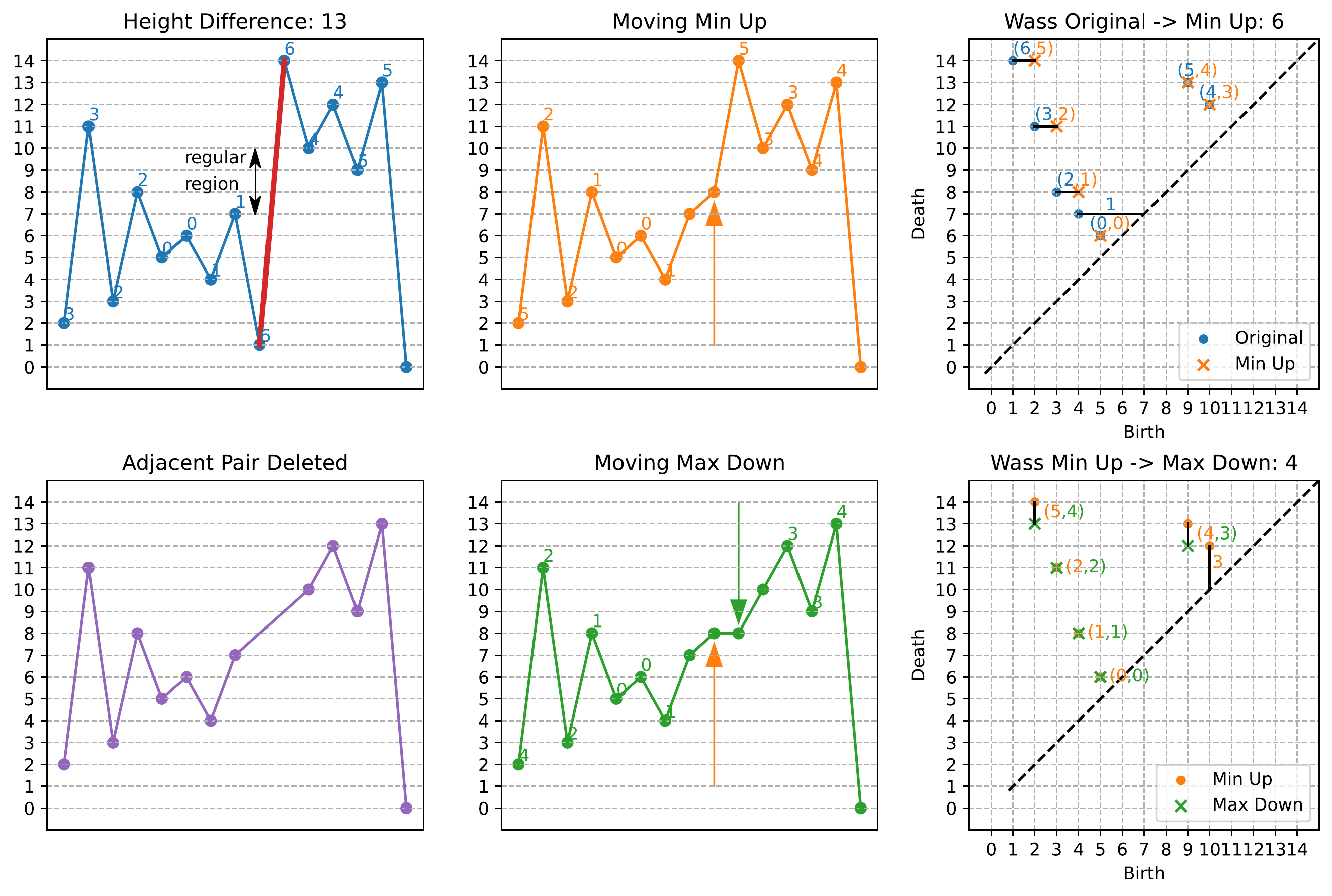}
\caption{Deleting a min-max pair is akin to changing the min's and max's heights. We delete the red min-max pair by moving the min up, and then by moving the max down. Multiple persistence pairing changes happen, but only one point moves at a time; the total L1 Wasserstein distance is at most the cost (height difference) we pay for the DOPE distance in general. In this example, the total Wasserstein distance along the path we track is $6+4=10$, which is $3$ less than the height difference. This is because, between heights $7$ and $10$, moving the min up from and the max down to height $8$ has no effect on the persistence diagrams, as the neighboring points have become regular.}
\label{fig:Informativity}
\end{figure}

\begin{proposition}
    Let $\vect{x}, \vect{y}$ be time series and $\mathrm{Dgm}(\vect{x}), \mathrm{Dgm}(\vect{y})$ the persistence diagrams for $\vect{x}^c$ and $\vect{y}^c$, respectively. Then:
    
    $\mathrm{d}_\mathcal{W}^1(\mathrm{Dgm}(\vect{x}),\mathrm{Dgm}(\vect{y})) \leq \dope{\vect{x}, \vect{y}}$
\end{proposition}

\begin{proof} We will show that performing any edit operation to $\vect{x}^c$ contributes at most the edit cost to the corresponding 1-Wasserstein distance between persistence diagrams. Intuitively, one can view edit operations as modifying the heights of critical points of $\vect{x}^c$. First observe that deleting a min-max pair is equivalent to raising the min and lowering the max until both become regular points (Figure \ref{fig:Informativity}). Similarly, a matching is equivalent to changing the height of a single point in $\vect{x}^c$. Since insertions into $\vect{x}^c$ are deletions from $\vect{y}^c$, all edit operations can be seen as modifying heights of points in the time series. Let $\vect{x}^c$ be a critical point time series and construct $\vect{\tilde{x}}^c$ by performing one edit to $\vect{x}^c$. Without loss of generality, suppose the edit raises a minimum $m$ in $\vect{x}^c$ by height $h$. This gives two cases for resulting changes in the persistence diagram of $\vect{x}^c$. Let $(b_m, d_m)$ be the birth-death pair corresponding to $m$.\\
\textbf{Case 1:} No persistence pairings change as a result of the edit. All birth-death pairs in $\mathrm{Dgm}(\vect{x}^c)$ and $\mathrm{Dgm}(\vect{\tilde{x}}^c)$ agree except for $(b_m, d_m)$ and $(b_m+h, d_m)$. Define a 1-Wasserstein matching by matching identical persistence pairs between $\vect{x}^c$ and $\vect{\tilde{x}}^c$ and matching $(b_m,d_m)$ in  $\mathrm{Dgm}(\vect{x}^c)$ to $(b_m+h,d_m)$ in $\mathrm{Dgm}(\vect{\tilde{x}}^c)$. The cost of this matching is $(m+h)-m = h$ and thus $\mathrm{d}_\mathcal{W}(\mathrm{Dgm}(\vect{x}^c), \mathrm{Dgm}(\vect{\tilde{x}}^c)) \leq h = \dope{\vect{x}^c, \vect{\tilde{x}}^c}$.\\
\textbf{Case 2:} Persistence pairings change as a result of the edit.
Without loss of generality, suppose only one change in persistence pairings occurs a result of $m$ surpassing the height of another minimum in $\vect{x}^c$ at height $b_m+h_1$, and let $\mathrm{Dgm}(\vect{\hat{x}}^c))$ be persistence diagram occuring when the height of $m$ is $b_m+h_1$. Since $m$ can only belong to one persistence pairing at a time, changing the height of $m$ only changes one birth-death pair at a time and we can accordingly partition the edit into two parts: raising $m$ by $h_1$, and raising it from $b_m+h_1$ to $b_m+h$. During both of these sub-edits put us back in Case 1. Note that if the change in persistence pairings at height $b_m+h_1$ results from $m$ becoming a regular point, the second sub-edit induces no change to the persistence diagrams. Thus $\mathrm{d}_\mathcal{W}(\mathrm{Dgm}(\vect{x}^c), \mathrm{Dgm}(\vect{\tilde{x}}^c)) = \mathrm{d}_\mathcal{W}(\mathrm{Dgm}(\vect{x}^c), \mathrm{Dgm}(\vect{\hat{x}}^c)) + \mathrm{d}_\mathcal{W}(\mathrm{Dgm}(\vect{\hat{x}}^c), \mathrm{Dgm}(\vect{\tilde{x}}^c)) \leq (m+h_1)-m + (m+h)-(m+h_1) = h = \dope{\vect{x}^c, \vect{\tilde{x}}^c}$. \\
Therefore, since moving a point in the time series until it becomes a regular point corresponds one-to-one with movement in the persistence diagram, regardless of whether this movement in the persistence diagram is split up across multiple points travelling along disjoint paths, for any single edit from $\vect{x}^c$ to $\vect{\tilde{x}}^c$, $\mathrm{d}_\mathcal{W}(\mathrm{Dgm}(\vect{x}^c), \mathrm{Dgm}(\vect{\tilde{x}}^c)) \leq \dope{\vect{x}^c, \vect{\tilde{x}}^c}$. Since $\dope{\vect{x}, \vect{y}}$ corresponds to a sequence of edits, the statement of the theorem then follows.\end{proof}
% TODO: Explain / reword the notion of "moving" points in the persistence diagram when what we really mean is computing the L1 cost for the SINGULAR bijection we've created between the persistence diagram before the height change and the persistence diagram after the height change.  
%%%%%%%%%%%%%%%%%%%%%%%%%%%%%%%%%%%%%%%_{i-k}

\subsection{Stability}

Claim: The DOPE distance is 1-stable by Definition~\ref{def:stable}.

\begin{proof}
    Let $M_c$ and $N_c$ be the lengths of $\vect{x^c}$ and $\vect{x^c}$, respectively, and zeropad them past their range as in Definition~\ref{def:stable}.  Then for $p=1$, we have
    
    \[ \sum_{i = 1}^{\max{M_c, N_c}} |\vect{x^c_i} - \vect{y^c_i}| = \left( \sum_{i=1}^{ \min{M_c, N_c}} |\vect{x^c_i} - \vect{y^c_i}|  \right) + \left( \sum_{i=M_c+1}^{N_c} |\vect{y^c_i}|  \right) + \left( \sum_{i=N_c+1}^{M_c} |\vect{x^c_i}|  \right)   \]
    where we take the convention that the sums are 0 if the upper index is less than the starting index. This sum upper bounds a possible dope matching.  The first term upper bounds a cost of the $L1$ matching the first $\min\{M_c, N_c\}$ points of $\vect{x^c}$ and $\vect{y^c}$.  The second term upper bounds the deletion cost of any parts of $\vect{y^c}$ that go beyond the range of $\vect{x^c}$, and the third term upper bounds the deletion cost of any parts of $\vect{x^c}$ that go beyond the range of $\vect{y^c}$.  Since we have shown how to construct a dope matching whose cost is upper bounded by the 1-stability definition, the optimal dope distance is also bounded.
    
    One case we have to consider is if the points deleted from $\vect{x^c}$ or $\vect{y^c}$ are odd in number, which would violate our rules of deleting in pairs.  However, since we're comparing time series on same domain, the Euler characteristic $\chi$ of their domains is the same.  This ensures that $(x_{\text{mins}} - x_{\text{maxes}}) = (y_{\text{mins}} - y_{\text{maxes}}) = \chi$, and so $(x_{\text{mins}} + x_{\text{maxes}}) - (y_{\text{mins}} + y_{\text{maxes}}) = 2 (x_{\text{maxes}} -  y_{\text{maxes}}) = M_c - N_c$, so $M_c-N_c$ is always even, and this will never happen.\end{proof}

%%%%%%%%%%%%%%%%%%%%%%%%%%%%%%%%%%%%%%%
\subsection{Efficient Computation}

Our algorithm to efficiently compute the DOPE distance uses the fact that critical points in a DOPE matching must match {\em in sequence}.  Let $d_{i,j}$ be the DOPE distance between the first $i$ critical points of a time series $\vect{x}$ and the first $j$ critical points of a time series $\vect{y}$ in sequence.  Then let $\vect{x}^c_i$ and $\vect{y}^c_j$ refer to the $i^\text{th}$ and $j^\text{th}$ critical points in $\vect{x}$ and $\vect{y}$, respectively, and let $\vect{x}^{cm}$ and $\vect{y}^{cm}$ be indicator functions for whether $\vect{x}^c$ and $\vect{y}^c$ are mins or maxes, i.e., $\vect{x}^{cm}_i = -1$ if $\vect{x}^c_i$ is a min and $\vect{x}^{cm}_i = 1$ if $\vect{x}^c_i$ is a max.  We also use the convention that an index of $0$ in $d$ refers to the emptyset, so that, for instance, $d_{0,j}$ would be the DOPE distance between the emptyset and the first $j$ critical points of $\vect{y}$.  Then the following recurrence holds:

\begin{lemma}
\label{lemma:recursivesubproblems}

\begin{equation}
    d_{i,j} =  \left\{
        \begin{array}{cc}
         0 & i = 0, j = 0 \\
        \\
         \displaystyle\sum_{\ell=1}^i \vect{x}^c_{\ell} \vect{x}^{cm}_{\ell} & i \Mod{2} = 0, j = 0 \\
        \\
         \displaystyle\sum_{\ell=1}^j \vect{y}^c_{\ell} \vect{y}^{cm}_{\ell} & i = 0, j \Mod{2} = 0 ,\\
        \\
         
        \min \left\{ 
        \begin{array}{c} 
        
            \left\{ 
                \begin{array}{cc}
                d_{i-1, j-1} + |\vect{x}^c_i - \vect{y}^c_j|  & \vect{x}^{cm}_i = \vect{y}^{cm}_j\\
                \infty & \text{otherwise}
                \end{array}
            \right\} \\
         \\ 
                        d_{i-2, j} + |\vect{x}^c_i - \vect{x}^c_{i-1}| \\
                        d_{i, j-2} + |\vect{y}^c_j - \vect{y}^c_{j-1}| \\
        \end{array}  \right\} & i, j \geq 2 \\
        \\
        
         \infty & \text{otherwise}

        \end{array}
    \right\}
\end{equation}

\end{lemma}

\begin{proof}
    The first three cases cover the boundary conditions, where the only option is to delete all critical points in pairs since one of the time series is empty.  The cost of this is the sum of max heights minus the sum of min heights.  The fourth condition covers the three general cases in the recurrence.  In the first such case, we consider matching the last time critical points $\vect{x}^c_i$ and $\vect{y}^c_j$ if they are either both mins or both maxes.  Then, the rest of the matching $d_{i-1, j-1}$ up to this point can be solved optimally independently of this choice.  The other two cases consider deleting the the min/max or max/min pairs at the end of either time series.  Since these are the only three ways to deal with the values at the end of $\vect{x}^c$ and $\vect{y}^c$, and the solutions of their respective subproblems are independent of their costs, taking the minimum of the three yields an optimal cost for $d_{i, j}$.
\end{proof}
A straightforward dynamic programming algorithm, much like the textbook algorithm to compute dynamic time warping and Levenshtein distance \cite{levenshtein1966binary}, follows from Lemma~\ref{lemma:recursivesubproblems}; more details can be seen in Appendix~\ref{sec:dynamicdope}.  Let $M$ and $N$ be the length of time series $\vect{x}$ and $\vect{y}$, respectively, and let $M_c$ and $N_c$ be the number of critical points on each time series, respectively.  Then the complexity of the algorithm is $O(M+N + M_c N_c)$.  Since $M_c \leq M$ and $N_c \leq N$, then the algorithm is $O(MN)$, matching the complexity of the textbook algorithm for dynamic time warping.  In practice, though, there will be many fewer critical points than overall points in the time series so this will be more efficient than dynamic time warping.  Notably, the bound is also significantly better than the $O((M_c+N_c)^3)$ bound used to compute a Wasserstein matching using the Hungarian algorithm \cite{kuhn1955hungarian}.

\section{Circular Domains}

\begin{figure}
\centering
\includegraphics[width=0.8\textwidth]{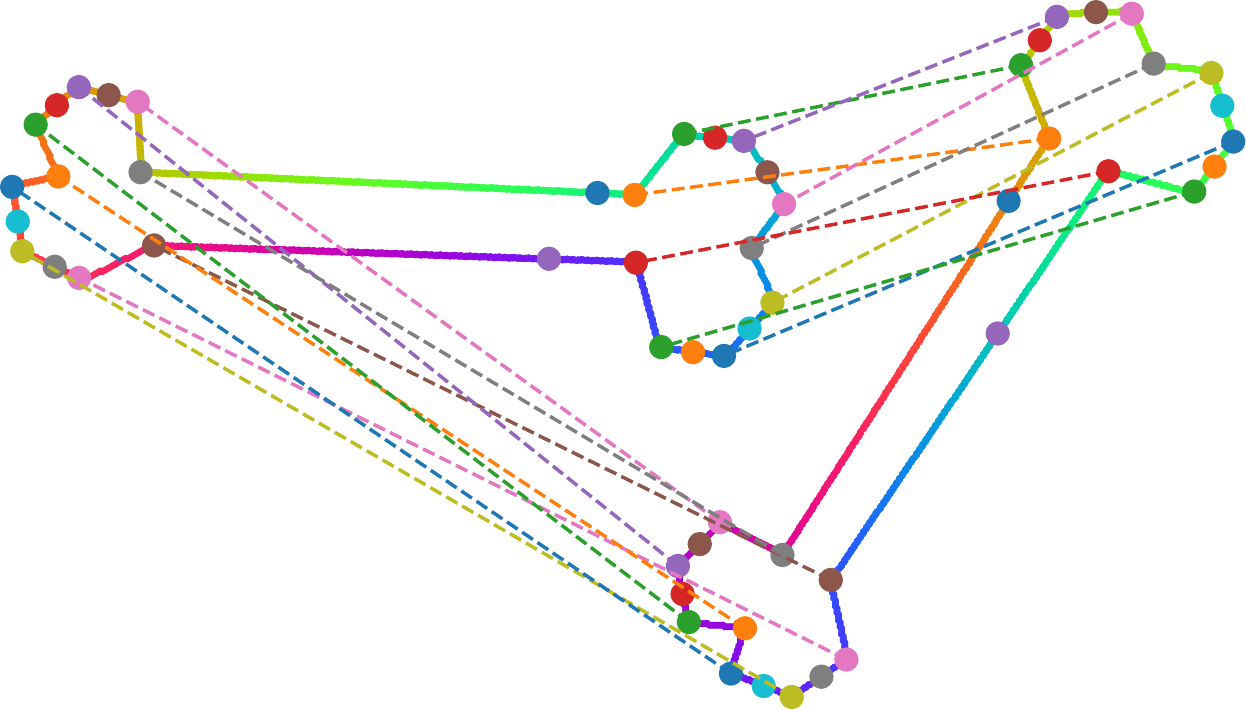}
\caption{An example of C-DOPE matching two bones from the mpeg-7 dataset.  Color indicates parameterization, and corresponding matched critical points of signed curvature are drawn as the same color.  C-DOPE correctly matches critical points, even under different parameterizations}
\label{fig:CircularDOPE_Example}
\end{figure}

We now extend DOPE to work on circular domains in addition to intervals.  Suppose we have two circular time series $\vect{x}$ and $\vect{y}$ (Definition~\ref{def:timeseries}):

\begin{definition}
    \label{def:cdope}
    The C-DOPE (Circular Dynamic Ordered Persistent Edit) distance between two circular time series $\vect{x}$ and $\vect{y}$ of with $M_c$ and $N_c$ critical points, respectively, is
    $
        \label{eq:cdope}
        \cdope{\vect{x}, \vect{y}} = \min_{i, j} \left( \dope{\overset{\rightarrow i}{\vect{x}}, \overset{\rightarrow j}{\vect{y}}} \right)
    $
over all $i = 1, 2, \hdots M_c$ and $j = 1, 2, \hdots N_c$.  

\end{definition}

In other words, C-DOPE is the minimum DOPE distance between all possible circularly shifted representative critical point time series of each equivalence class of circularly shifted time series.  Like DOPE, C-DOPE also satisfies stability, informativity, and metric properties.  The proofs of these are very similar, so we do not repeat them here.  We can compute C-DOPE naively in $O(M + N + M_c^2 N_c^2)$ time directly from the definition, though it is more efficient to hold one representative fixed and to shift the other:

\begin{lemma}
    $
        \cdope{\vect{x}, \vect{y}} = \min \left(  \min_j (\dope{\vect{x}, \overset{\rightarrow j}{\vect{y}}}), \min_j (\dope{\overset{\rightarrow 1}{\vect{x}}, \overset{\rightarrow j}{\vect{y}}})  \right)
    $
\end{lemma}

The proof of this follows from Definition~\ref{def:cdope}, and such a scheme is cubic $O(M_c N_c^2)$.  The only subtlety is that we need to do this twice: once holding $\vect{x}$ fixed and shifting $\vect{y}$, and once after circularly shifting $\vect{x}$ by one and shifting $\vect{y}$.  Doing this a second time on $\overset{\rightarrow 1}{\vect{x}}$ takes care of the case where an optimal solution deletes the min/max pair occurring in $\vect{x}_0$ and $\vect{x}_{M_c-1}$.

As an example, we compute smoothed curvature from two bone contours from the MPEG-7 database \cite{latecki2000shape} using the technique of Mokhtarian and Mackworth \cite{mokhtarian1992theory}.  Figure~\ref{fig:CircularDOPE_Example} shows the result.  As this example shows, since curvature is an isometry invariant, and since C-DOPE is blind to parameterization, C-DOPE can be used to match loops that have been rotated, translated, and re-parameterized.  We explore this in depth in Section~\ref{sec:experimentcdope}.

\section{Experiments}
\label{sec:experiments}
We now empirically examine the performance of DOPE and some related algorithms on a variety of classification tasks in real data.  We focus on algorithms that, like DOPE, are both {\em parameter free} and {\em unsupervised}; that is, they can provide a similarity measure with no training data.  Following the work of \cite{dau2019ucr}, we evaluate the performance of each dataset on a particular method by holding out each example and ranking the other examples according to a chosen similarity measure.  Let $\ell$ be the class label of a time series $\vect{x}$, and let $\vect{c} = [\vect{c}_1, \vect{c}_2, \hdots \vect{c}_N]$ be the class labels of the rest of the time series in decreasing order of similarity (e.g. increasing DOPE distance).  If the similarity measure captures class membership appropriately, the first element is the most likely to be in the same class as $\vect{x}$, and the last element is the least likely to be in the same class as $\vect{x}$.  To quantify this, use mean rank (MR), as is standard in UCR comparisons.  However, as the UCR authors note, mean rank can be misleading and paradoxical over many datasets in practice \cite{benavoli2016should}, so we also report the \textbf{Mean Average Precision (MAP)}, which is more robust to outliers items in a class that are ranked very late compared to the others.  See Appendix~\ref{sec:mrmap} for more details.

\subsection{UCR Time Series Dataset}

\begin{figure}[h]
\centering
\includegraphics[width=\textwidth]{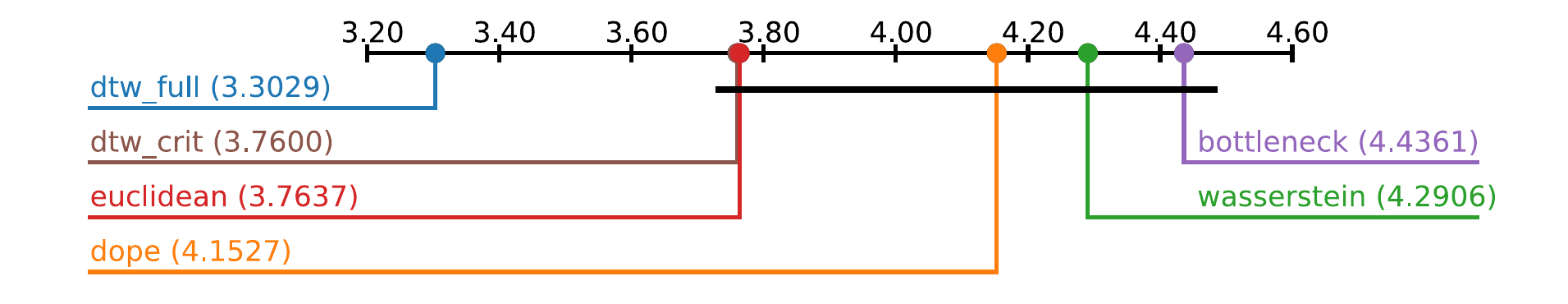}
\caption{Mean rank (MR) on the UCR time series dataset.  Lower values are better}
\label{fig:UCRMR}
\end{figure}

\begin{figure}
\centering
\includegraphics[width=\textwidth]{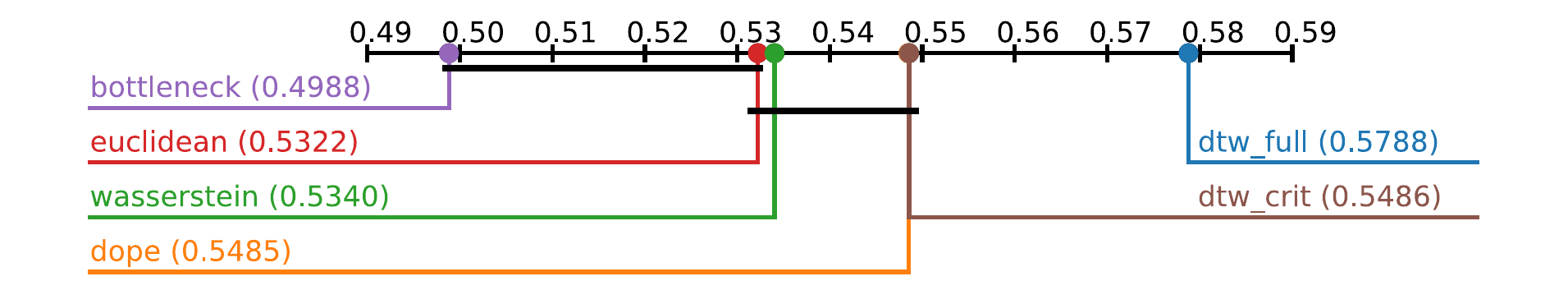}
\caption{Mean Average Precision (MAP) on the UCR time series dataset.  Higher is better}
\label{fig:UCRMAP}
\end{figure}

We first evaluate several time series techniques on the UCR Time Series database \cite{dau2019ucr}, which consists of 128 different time series classification across a wide variety of application domains.  We evaluate MR and MRR on DOPE, Bottleneck/Wasserstein distance, point-by-point Euclidean distance (with zeropadding for unequal length signals), and DTW.  We also report results of DTW on the critical point time series, which is similar to DOPE, but which still does not satisfy metric properties or stability (Appendix~\ref{sec:appendixdtwpropeties}).  Following recommendations from the UCR authors, we create ``critical distance plots'' to summarize the MR (Figure~\ref{fig:UCRMR}) and MAP (Figure~\ref{fig:UCRMAP}) statistics across all datasets.  Similarity technique are grouped together into cliques using pairwise Wilcoxon signed-rank tests \cite{demvsar2006statistical}.  Pairs of similarity measurements with a Holm-corrected \cite{holm1979simple} Wilcoxon $p$-value under 0.05 are grouped together with a line.  As the results show, DOPE performs better than the bottleneck and Wasserstein distances, and similarly to DTW on critical point time series (though we know that DOPE has better theoretical properties).  Interestingly, though, DTW performs better than all of the above, which suggests that in practice, important class information may be contained in parameterizations.  Thus, one should always try DTW and other simple off-the-shelf methods first before ruling out parameterization as a nuisance.

\begin{figure}
\centering
\includegraphics[width=\textwidth]{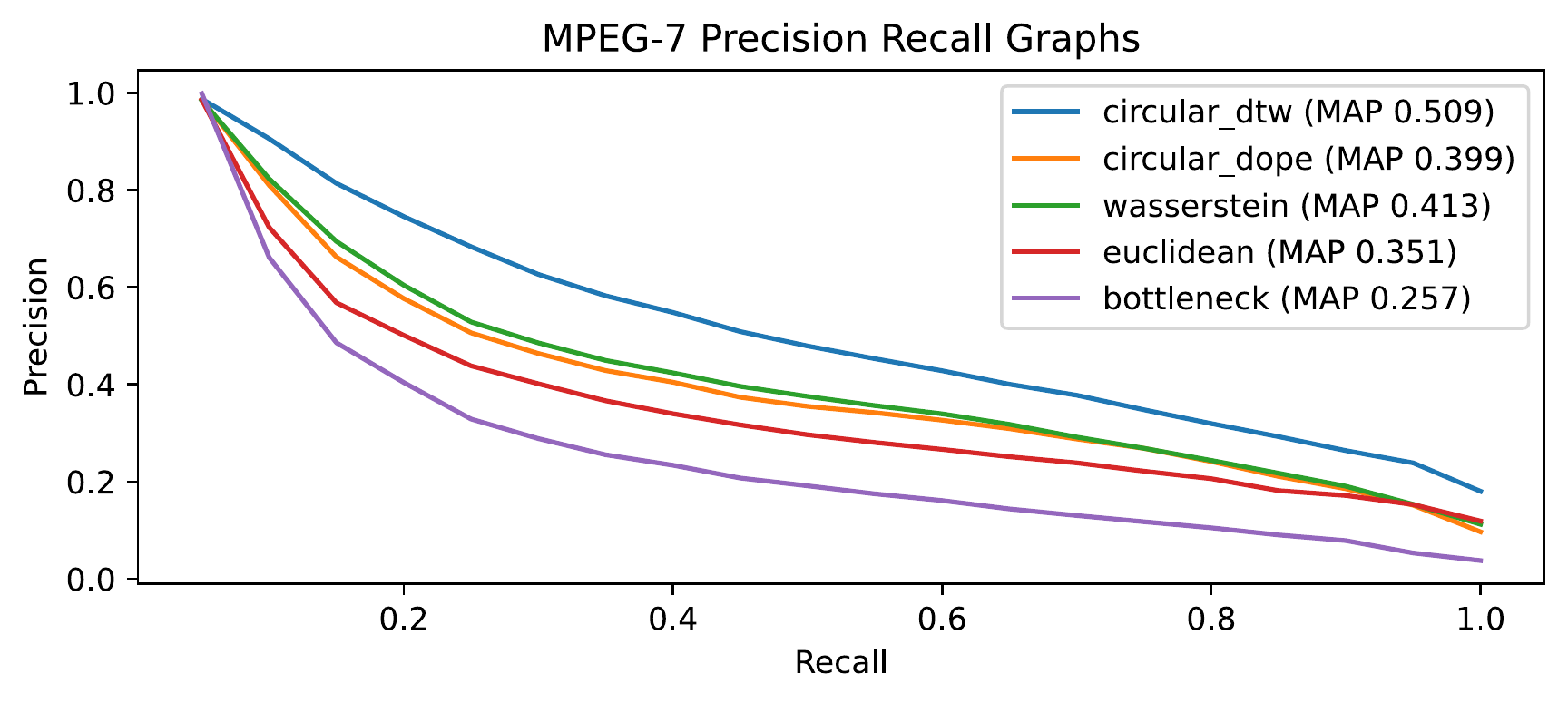}
\caption{Precision-recall curves on the MPEG-7 dataset}
\label{fig:mpeg7}
\end{figure}

\subsection{MPEG-7 Shape Contours Dataset}
\label{sec:experimentcdope}
We also evaluate C-DOPE individually on the MPEG-7 dataset, which consists of binary images of 72 different shape classes (e.g. apple, car, deer, octopus), each with 20 shapes.  We convert each image to a closed loop by using marching squares to extract the largest boundary component between the foreground and background.  We then compute the signed curvature as our time series value at each pixel using the technique of Mokhtarian/Makworth \cite{mokhtarian1992theory}.  Figure~\ref{fig:CircularDOPE_Example} shows an example of two such loops from the ``bone'' class.  We then compare all pairwise loops using C-DOPE, Bottleneck/Wasserstein distance between persistence diagrams filtered on the circular domain, and circular dynamic time warping \cite{marzal2005dynamic}.  Figure~\ref{fig:mpeg7} shows the resulting precision recall curves, which overall match the same trends as the UCR time series experiment, except Wasserstein does slightly better than DOPE for low recall.

\subsection{Experiment Discussion}
\label{sec:discussion}

%In this work, we presented a time series classification method, dubbed ``DOPE,'' that was inspired by desired theoretical properties of ordered merge trees.  By restricting to ordered domains, DOPE satisfies stability, informativity, and metric properties, all while being computable in quadratic time.  To our knowledge, this is the first metric specialized to merge trees that satisfies these properties.  Not only is this faster and more informative than other computable topological signatures such as the persistence diagram Wasserstein distance, but it also performs better in practice on a large database of time series data (Section~\ref{sec:experiments}), likely because it can tell the difference between reflections local max motions (Figure~\ref{fig:BlindRotation}).

Interestingly, while we believe that DOPE is uniquely useful at performing a sparse alignment between warped time series (e.g. Figure~\ref{fig:ECGConceptDOPE} and Figure~\ref{fig:CircularDOPE_Example}), DOPE does not reach the classification performance of off-the-shelf DTW on the datasets we examined, so classification may not be the best application of DOPE.  Overall, more stable techniques sacrifice informativity, and this can even be seen in the consistent superior performance of Wasserstein distance over Bottleneck distance, even though the latter is a stronger $L_{\infty}$-stable.  Regardless, this should serve as a cautionary tale for those seeking to apply topological techniques; one should always check assumptions about how important certain theoretical properties are in practice.

\bibliographystyle{plain}
\bibliography{main}

\section{Appendix}

\subsection{Dynamic Time Warping: Definition And Examples}
\label{sec:dtwappendix}

We provide more information here on dynamic time warping, including counter-examples to the triangle inequality and stability.  Central to the definition of DTW is the notion of a {\em warping path} $\mathcal{W}$ between two time series:

\begin{definition}
\label{def:warppath}
A {\em warping path} $\mathcal{W}$ is a sequence of pairs between the index sets of two time series, with $M$ and $N$ samples, respectively, satisfying
\begin{itemize}
    \item $(0, 0), (M-1, N-1) \in \mathcal{W}$; that is, a warping path starts at the beginning of both time series and ends at the end of both
    
    \item $W_i - W_{i-1} \in \left\{ \begin{array}{c} (1, 0) \\ (0, 1) \\ (1, 1) \end{array} \right\}$; that is, subsequent paired indices differ by at most one, but at least one index must advance forward
\end{itemize}
\end{definition}

The dynamic time warping similarity between two 1D time series $\vect{x}$ and $\vect{y}$ is then defined as 

\begin{definition}
    \label{def:dtw}
    \[ \text{DTW}(\vect{x}, \vect{y}) = \min_{\mathcal{W}} \left\{\sum_{(i, j) \in \mathcal{W}} |\vect{x}_i - \vect{y}_j|\right \} \]
\end{definition}

Dynamic time warping can be computed efficiently using dynamic programming.  Let $D_{ij}$ be the subproblem of aligning the first $i$ points of $\vect{x}$ to the first $j$ points of $\vect{y}$.  Then the following recurrence holds:

\begin{equation}
    \text{DTW}_{ij} = |\vect{x}_i - \vect{y}_j| + \min \left\{ \begin{array}{c}  \text{DTW}_{i-1, j} \\ \text{DTW}_{i, j-1} \\ \text{DTW}_{i-1, j-1} \end{array} \right\}
\end{equation}
with base conditions $\text{DTW}_{1j} = \displaystyle\sum_{n = 0}^j |\vect{x}_0 - \vect{y}_j|$ and $\text{DTW}_{i0} = \displaystyle \sum_{m = 0}^i |\vect{x}_i - \vect{y}_1|$.  A naive application of these recurrences yields an $O(MN)$ algorithm, though it has been shown that it is possible to compute this distance in time $O(N^2 \log \log \log(N) / \log \log(N)$ for $N$ as $O(M)$ for 1D time series \cite{gold2018dynamic}.  

\subsubsection{Triangle Inequality And Stability}
\label{sec:appendixdtwpropeties}
Unfortunately, dynamic time warping does not satisfy the triangle inequality.  Consider the following family of three time series:

\begin{enumerate}
    \item $\vect{a} = [-1, -1, ..., -1, 0]$, where $-1$ is repeated $m > 1$ times
    \item $\vect{b} = [-1, 0, 1]$
    \item $\vect{c} = [0, 1, 1, ..., 1]$, where $1$ is repeated $n > 1$ times
\end{enumerate}

Then $\text{DTW}(\vect{a}, \vect{b}) = 1$, $\text{DTW}(\vect{b}, \vect{c}) = 1$, and $\text{DTW}(\vect{a}, \vect{c}) = m + n$.  Thus, $\text{DTW}(\vect{a}, \vect{c}) > \text{DTW}(\vect{a}, \vect{b}) + \text{DTW}(\vect{b}, \vect{c})$.  Intuitively, something has gone wrong when a warping path is forced to match many of the same element in a row.  A method based on critical points, on the other hand, would collapse such a run of values.

Dynamic time warping also violates stability.  Consider two time series that are samples of the exact same function, and suppose they take the exact same samples except within a monotonic region that goes from $0$ to $1$.  Let $\vect{x}$ just take these two samples, so $\vect{x} = [\vect{a}, 0, 1, \vect{b}]$ (where $\vect{a}$ and $\vect{b}$ represent arbitrary sub time series), but let $\vect{y}$ take $n+1$ samples equally spaced from $0$ to $1$ over that interval, so $\vect{y} = [\vect{a}, 1/n, 2/n, \ldots , (n-1)/n, 1, \vect{b}]$. Then 
\[\text{DTW}(\vect{x}, \vect{y}) \approx 2 \sum_{i = 1}^{n/2} i/n = n/4 + 1\]
There are no critical points between $a$ and $b$, so any nonzero distance will violate stability.  Intuitively, the problem is that we have sampled many more regular points in $\mathbf{y}$ than in $\mathbf{x}$.  Since a merge tree only tracks critical points, any metric based off of a merge tree will avoid problems in this particular example.

A natural followup question is whether DTW on a critical point time series can ameliorate the issues that these two examples expose.  Unfortunately, this is not the case either.  Consider the following critical point time series

\begin{enumerate}
    \item $\vect{a} = [-1-\epsilon, -1, -1-\epsilon, -1 ..., -1-\epsilon, 0, -1]$, length($\vect{a}$) = 2m + 1
    \item $\vect{b} = [-1, 1, -1]$
    \item $\vect{c} = [0, 1+\epsilon, 1, 1+\epsilon, 1, \hdots, 1+\epsilon, 1]$, length($\vect{c}$) = 2n + 1
\end{enumerate}

then $\text{DTW}(\vect{a}, \vect{b}) = 1 + m \epsilon$, $\text{DTW}(\vect{b}, \vect{c}) = 3 + n \epsilon$, and $\text{DTW}(\vect{a}, \vect{c}) = (m+n) (2 + \epsilon) $.  Once again, this example violates the triangle inequality for all $m, n \geq 1$.

It is worth noting that in practice, as the authors of \cite{mueen2016} show, it is more likely that DTW will satisfy the triangle inequality if the warping paths are restricted to the so-called ``Sakoe-Chiba band'' \cite{sakoe1978dynamic}; that is, $|i-j| < c$ for all $(i, j) \in \mathcal{W}$.

\subsection{Details of the proof of the triangle inequality for DOPE}
\label{sec:dopetriangleappendix}
Details for proof of Lemma \ref{lem: triangle edit operations} giving a triangle inequality on individual edits:
\begin{proof}
We will prove the triangle inequality in the cases where $p$ and $r$ are min-max pairs, $p=(\vect{x}^c_i,\vect{x}^c_{i+1})$, $r=(\vect{z}^c_k,\vect{z}^c_{k+1})$.  There are two such cases corresponding to the two paths from $p$ to $r$. Let Path 1 be the path formed when $q$ is a min-max pair $q=(\vect{y}_j^c, \vect{y}_{j+1}^c)$, and let Path 2 be the path formed by $q=\Lambda$. Then Path 1 consists of matching min-max pair $p=(\vect{x}_i^c, \vect{x}_{i+1}^c)$ with $(\vect{y}_j^c, \vect{y}_{j+1}^c)$ followed by matching $(\vect{y}_j^c, \vect{y}_{j+1}^c)$ with $(\vect{z}_k^c, \vect{z}_{k+1}^c)$. Contrastingly, Path 2 involves deleting $(\vect{x}_i^c, \vect{x}_{i+1}^c)$ followed by inserting $(\vect{z}_k^c, \vect{z}_{k+1}^c)$.  For Path 1, we wish to show that:
    \begin{align*}
        \gamma((\vect{x}_i^c, \vect{x}_{i+1}^c) \rightarrow (\vect{z}_k^c, \vect{z}_{k+1}^c)) \leq \gamma((\vect{x}_i^c, \vect{x}_{i+1}^c) \rightarrow (\vect{y}_j^c, \vect{y}_{j+1}^c)) + \gamma ((\vect{y}_j^c, \vect{y}_{j+1}^c) \rightarrow (\vect{z}_k^c, \vect{z}_{k+1}^c))
    \end{align*}
    This requires showing that:
    \begin{align*}
        \min{\begin{cases}
        |\vect{x}^c_i-\vect{z}^c_k| + |\vect{x}^c_{i+1}-\vect{z}^c_{k+1}| \\
        |\vect{x}^c_i-\vect{x}^c_{i+1}| + |\vect{z}^c_k-\vect{z}^c_{k+1}| \\
        \end{cases}}
        \leq
        \min{\begin{cases}
        |\vect{x}^c_i-\vect{y}^c_j| + |\vect{x}^c_{i+1}-\vect{y}^c_{j+1}| \\
        |\vect{x}^c_i-\vect{x}^c_{i+1}| + |\vect{y}^c_j-\vect{y}^c_{j+1}| \\
        \end{cases}}
        +
        \min{\begin{cases}
        |\vect{y}^c_j-\vect{z}^c_k| + |\vect{y}^c_{j+1}-\vect{z}^c_{k+1}| \\
        |\vect{y}^c_j-\vect{y}^c_{j+1}| + |\vect{z}^c_k-\vect{z}^c_{k+1}| \\
        \end{cases}}
    \end{align*}
    To do so, we will examine each of the four resulting cases separately.
    \begin{align*}
    \gamma(p \rightarrow r) &= \gamma((\vect{x}_i,\vect{x}_{i+1}) \rightarrow (\vect{z}_k,\vect{z}_{k+1})) \\
    & = \min{\begin{cases}
    |\vect{x}_i-\vect{z}_k| + |\vect{x}_{i+1}-\vect{z}_{k+1}| \\
    |\vect{x}_i-\vect{x}_{i+1}| + |\vect{z}_k-\vect{z}_{k+1}| \\
    \end{cases}}
    \end{align*}
    
    CASE 1:
    \begin{align*}
    & \min{\begin{cases}
    |\vect{x}_i-\vect{z}_k| + |\vect{x}_{i+1}-\vect{z}_{k+1}| \\
    |\vect{x}_i-\vect{x}_{i+1}| + |\vect{z}_k-\vect{z}_{k+1}| \\
    \end{cases}} \leq |\vect{x}_i-\vect{z}_k| + |\vect{x}_{i+1}-\vect{z}_{k+1}| \\
    & = |\vect{x}_i-\vect{y}_j+\vect{y}_j-\vect{z}_k| + |\vect{x}_{i+1}-\vect{y}_{j+1}+\vect{y}_{j+1}-\vect{z}_{k+1}| \\
    & \leq |\vect{x}_i-\vect{y}_j| + |\vect{y}_j-\vect{z}_k| + |\vect{x}_{i+1}-\vect{y}_{j+1}| +|\vect{y}_{j+1}-\vect{z}_{k+1}| \\
    & = \left( |\vect{x}_i-\vect{y}_j| + |\vect{x}_{i+1}-\vect{y}_{j+1}| \right) + \left( |\vect{y}_j-\vect{z}_k| + |\vect{y}_{j+1}-\vect{z}_{k+1}| \right) \\
    \end{align*}

    CASE 2:
    \begin{align*}
    & \min{\begin{cases}
    |\vect{x}_i-\vect{z}_k| + |\vect{x}_{i+1}-\vect{z}_{k+1}| \\
    |\vect{x}_i-\vect{x}_{i+1}| + |\vect{z}_k-\vect{z}_{k+1}| \\
    \end{cases}} \leq |\vect{x}_i-\vect{x}_{i+1}| + |\vect{z}_k-\vect{z}_{k+1}| \\
    & = |\vect{x}_i-\vect{y}_j+\vect{y}_j-\vect{y}_{j+1}+\vect{y}_{j+1}-\vect{x}_{i+1}| + |\vect{z}_k-\vect{z}_{k+1}|\\
    & \leq |\vect{x}_i-\vect{y}_j|+|\vect{y}_j-\vect{y}_{j+1}|+|\vect{y}_{j+1}-\vect{x}_{i+1}| + |\vect{z}_k-\vect{z}_{k+1}| \\
    & = \left(|\vect{x}_i-\vect{y}_j|+ |\vect{x}_{i+1}-\vect{y}_{j+1}| \right) + \left(|\vect{y}_j-\vect{y}_{j+1}|+|\vect{z}_k-\vect{z}_{k+1}| \right)\\
    \end{align*}

    CASE 3:
    \begin{align*}
    & \min{\begin{cases}
    |\vect{x}_i-\vect{z}_k| + |\vect{x}_{i+1}-\vect{z}_{k+1}| \\
    |\vect{x}_i-\vect{x}_{i+1}| + |\vect{z}_k-\vect{z}_{k+1}| \\
    \end{cases}} \leq |\vect{x}_i-\vect{x}_{i+1}| + |\vect{z}_k-\vect{z}_{k+1}| \\
    & = |\vect{x}_i-\vect{x}_{i+1}| + |\vect{z}_k-\vect{y}_j+\vect{y}_j-\vect{y}_{j+1}+\vect{y}_{j+1}-\vect{z}_{k+1}|\\
    & \leq |\vect{x}_i-\vect{x}_{i+1}| + |\vect{z}_k-\vect{y}_j|+|\vect{y}_j-\vect{y}_{j+1}|+|\vect{y}_{j+1}-\vect{z}_{k+1}|\\
    & = \left(|\vect{x}_i-\vect{x}_{i+1}| + |\vect{y}_j-\vect{y}_{j+1}|\right) + \left(|\vect{y}_j-\vect{z}_k|+|\vect{y}_{j+1}-\vect{z}_{k+1}| \right)\\
    \end{align*}

    CASE 4:
    \begin{align*}
    & \min{\begin{cases}
    |\vect{x}_i-\vect{z}_k| + |\vect{x}_{i+1}-\vect{z}_{k+1}| \\
    |\vect{x}_i-\vect{x}_{i+1}| + |\vect{z}_k-\vect{z}_{k+1}| \\
    \end{cases}} \leq |\vect{x}_i-\vect{x}_{i+1}| + |\vect{z}_k-\vect{z}_{k+1}| \\
    & \leq \left(|\vect{x}_i-\vect{x}_{i+1}| + |\vect{y}_j-\vect{y}_{j+1}|\right) + \left(|\vect{y}_j-\vect{y}_{j+1}| + |\vect{z}_k-\vect{z}_{k+1}|\right)
    \end{align*}
    
    Therefore:
    \begin{align*}
    & \min{\begin{cases}
        |\vect{x}^c_i-\vect{z}^c_k| + |\vect{x}^c_{i+1}-\vect{z}^c_{k+1}| \\
        |\vect{x}^c_i-\vect{x}^c_{i+1}| + |\vect{z}^c_k-\vect{z}^c_{k+1}| \\
        \end{cases}}
        \leq
        \min{\begin{cases}
        |\vect{x}^c_i-\vect{y}^c_j| + |\vect{x}^c_{i+1}-\vect{y}^c_{j+1}| \\
        |\vect{x}^c_i-\vect{x}^c_{i+1}| + |\vect{y}^c_j-\vect{y}^c_{j+1}| \\
        \end{cases}}
        +
        \min{\begin{cases}
        |\vect{y}^c_j-\vect{z}^c_k| + |\vect{y}^c_{j+1}-\vect{z}^c_{k+1}| \\
        |\vect{y}^c_j-\vect{y}^c_{j+1}| + |\vect{z}^c_k-\vect{z}^c_{k+1}| \\
        \end{cases}} \\
    & = \gamma((\vect{x}_i,\vect{x}_{i+1}) \rightarrow (\vect{y}_j, \vect{y}_{j+1})) + \gamma ((\vect{y}_j, \vect{y}_{j+1})\rightarrow (\vect{z}_k, \vect{z}_{k+1})) \\
    & = \gamma (p \rightarrow q) +\gamma (q \rightarrow r)
    \end{align*}
    so the triangle inequality holds for Path 1.\\
    
    For Path 2 with $q=\Lambda$, we have
    \begin{align*}
    \gamma(p \rightarrow r) &= \gamma((\vect{x}_i, \vect{x}_{i+1}) \rightarrow (\vect{z}_k,\vect{z}_{k+1})) \\\\
    & = \min{\begin{cases}
    |\vect{x}_i-\vect{z}_k| + |\vect{x}_{i+1}-\vect{z}_{k+1}| \\
    |\vect{x}_i-\vect{x}_{i+1}| + |\vect{z}_k-\vect{z}_{k+1}| \\
    \end{cases}}  \\\\ 
    & \leq |\vect{x}_i-\vect{x}_{i+1}| + |\vect{z}_k-\vect{z}_{k+1}| \\\\
    & = \gamma((\vect{x}_i,\vect{x}_{i+1}) \rightarrow \Lambda) + \gamma (\Lambda \rightarrow (\vect{z}_k, \vect{z}_{k+1})) \\\\
    & = \gamma (p \rightarrow q) +\gamma (q \rightarrow r).\\
    \end{align*}

    The other six cases are proved either trivially or in similar fashion to what is shown.
    \end{proof}

We next describe how to form an alignment given a sequence of edits.
\begin{lemma} \label{dopeleqedits}
For a finite sequence of edit operations $S=s_1, s_2, ..., s_n$ converting time series $\vect{x}^c$ into time series $\vect{y}^c$, 
$$\dope{\vect{x}, \vect{y}} \leq \gamma(S) = \sum_{k=1}^{n} \gamma (s_k)$$
\end{lemma}

\begin{proof}
Each edit $s_k$ is either an insertion, deletion, or pairing of a point in $\vect{x}^c$ with a point in $\vect{y}^c$. View the insertion of the pair $(a,b)$ into $\vect{x}^c$ as the deletion of the pair $(a,b)$ from $\vect{y}$. This pair must exist in $\vect{y}^c$ since the only allowed insertions are min-max pairs from the target time series. Form the removal sets $R_{\vect{x}^c}$ and  $R_{\vect{y}^c}$. The pairing operations $s \in S$ form the matching set $M$. This yields an alignment $A = (M,R_{\vect{x}^c},R_{\vect{y}^c})$ between time series $\vect{x}$ and $\vect{y}$. Since $\dope{\vect{x}, \vect{y}}$ is by definition the infimum alignment cost over all such alignments, we see that $\dope{\vect{x}, \vect{y}} \leq C(A)$. Hence
\begin{align*}
    C(A) &= \displaystyle\sum_{(\vect{x}^c_i, \vect{y}^c_j) \in M} |\vect{x}^c_i - \vect{y}^c_j| + \displaystyle\sum_{(\vect{x}^c_i, \vect{x}^c_{i+1}) \in R_{\vect{x}^c}} |\vect{x}^c_i - \vect{x}^c_{i+1}| + \displaystyle\sum_{(\vect{y}^c_j, \vect{y}^c_{j+1}) \in R_{\vect{y}^c}} |\vect{y}^c_j - \vect{y}^c_{j+1}| \\\\
    &= \displaystyle\sum_{(\vect{x}^c_i, \vect{y}^c_j) \in M} \gamma(\vect{x}^c_i \rightarrow \vect{y}^c_j) + \sum_{(\vect{x}^c_i, \vect{x}^c_{i+1}) \in R_{\vect{x}^c}} \gamma((\vect{x}^c_i, \vect{x}^c_{i+1}) \rightarrow \Lambda) + \sum_{(\vect{y}^c_i, \vect{y}^c_{i+1}) \in R_{\vect{y}^c}} \gamma((\vect{y}^c_i, \vect{y}^c_{i+1}) \rightarrow \Lambda) \\\\
    &= \displaystyle\sum_{s \in S} \gamma (s)\\\\
    &= \displaystyle\sum_{k=1}^{n} \gamma (s_k)\\
\end{align*}
so $$\dope{\vect{x}, \vect{y}} \leq \displaystyle\sum_{k=1}^{n} \gamma (s_k).$$
\end{proof}

Finally, we provide details of proof of Lemma \ref{compositionlemma}: 
\begin{proof}
We can intuitively visualize the described resolutions of overlap in Figure 6 as collapsing all insertions and deletions to leave only the points in $\vect{x}^c$ and $\vect{z}^c$. We now provide the explicit computaion for the even and odd overlap cases. Without loss of generality, we consider the simplest examples of even overlap (2 insertions or deletions) and odd overlap (3 insertions or deletions), and use simplified rather than generif indices as well.\\
    Even Case:
    \begin{align*} 
    \gamma(x \rightarrow z) = |x-z| &= |x-y_1+y_1-y_2+y_2-y_3+y_3-z| \\
    &\leq |x-y_1|+|y_1-y_2|+|y_2-y_3|+|y_3-z| \\
    &= \gamma(x \rightarrow y_1) + \gamma(\Lambda \rightarrow (y_2, y_3)) + \gamma((y_1,y_2) \rightarrow \Lambda) + \gamma(y_3 \rightarrow z)    
    \end{align*}
    Odd Case:
    \begin{align*} 
    \gamma(\Lambda \rightarrow (z_1,z_2)) &= |z_1-z_2| = |z_1-y_1+y_1-y_2+y_2-y_3+y_3-y_4+y_4-z| \\
    &\leq |z-y_1|+|y_1-y_2|+|y_2-y_3|+|y_3-y_4|+|y_4-z| \\
    &= \gamma(y_1 \rightarrow z_1) + \gamma(\Lambda \rightarrow (y_1, y_2)) + \gamma((y_2,y_3) \rightarrow \Lambda) + \gamma(\Lambda \rightarrow (y_3, y_4)) + \gamma(y_4 \rightarrow z_2)    
    \end{align*}\\
\end{proof}

\subsection{DOPE Dynamic Programming Algorithm}
\label{sec:dynamicdope}
We present in detail the algorithm to compute DOPE, and we annotate the time complexity of each step.  Let $M$ and $N$ be the length of time series $\vect{x}$ and $\vect{y}$, respectively, and let $M_c$ and $N_c$ be the number of critical points on each time series, respectively.  The algorithm proceeds by filling in all subproblems in an $(M_c+1) \times (N_c+1)$ dynamic programming table as follows

\begin{enumerate}
    \item $O(M+N)$: Compute the critical point time series $\vect{x}_c$ and $\vect{y}_c$ from $\vect{x}$ and $\vect{y}$, respectively
    \item $O(M_c + N_c)$: Fill in the base conditions $d_{0, 0}$ and $d_{i, 0}$ and $d_{0, j}$ (using $\infty$ for $i$ and $j$ odd).
    \item $O(M_c N_c)$: Fill in the rest of the table in lexicographic order on the tuples $(i, j)$, computing the min over the matching or two deletion possibilities.  This order ensures that all subproblems are computed before they are needed.  Since each entry $d_{i,j}$ can be computed in constant time given answers to the appropriate subproblems.
\end{enumerate}

\subsection{Rank Evaluation Statistics} 
\label{sec:mrmap}

\textbf{Mean rank (MR)} is defined as follows
\begin{definition}
    \[ MR = \left( \sum_{i=1}^N \left\{  \begin{array}{cc} i & c_i = \ell \\ 0 & \text{otherwise} \end{array} \right\}  \right) / \left( \sum_{i=1}^N \left\{  \begin{array}{cc} 1 & c_i = \ell \\ 0 & \text{otherwise} \end{array} \right\}  \right) \]
\end{definition}
In other words, the MR is the average rank of items that are in the same class as $\vect{x}$.  This is the go-to measure for the UCR time series dataset \cite{dau2019ucr}, and it is reported as an average of the MR over all examples.  

\textbf{Mean Average Precision (MAP)} is based on notions of \textbf{precision (P)} and \textbf{recall (R)}, each of which is defined up to a particular index $i$ in the ranked list:

\begin{definition}
$P_i = \left( \sum_{j=1}^i \left\{  \begin{array}{cc} j & c_j = \ell \\ 0 & \text{otherwise} \end{array} \right\} \right) / i$
\end{definition}

\begin{definition}
    $R_i = \left( \sum_{j=1}^i \left\{  \begin{array}{cc} 1 & c_j = \ell \\ 0 & \text{otherwise} \end{array} \right\}  \right) / \left( \sum_{j=1}^N \left\{  \begin{array}{cc} 1 & c_j = \ell \\ 0 & \text{otherwise} \end{array} \right\}  \right)$
\end{definition}

Intuitively, precision is a measure of the proportion of correct examples up to a certain point in the ranked list, while recall is a measure of the total examples in the same class up to a certain point.  There is an inherent trade-off between the two, and one often sees downward sloping curves when plotting ``precision-recall curves'' (e.g. Figure~\ref{fig:mpeg7}).  One way to summarize these curves with a single number is to simply take the average of all precisions (AP) over all unique recalls (i.e. the area under the precision recall curves).  Then, averaging the AP over all examples in the dataset gives us the MAP.

\end{document}